\documentclass[journal]{IEEEtran}
\newcommand{\RomanNumeralCaps}[1]
    {\MakeUppercase{\romannumeral #1}}
\usepackage{colortbl} 
\usepackage{amsmath,amsfonts}
\usepackage{algorithmic}
\usepackage[linesnumbered, ruled]{algorithm2e}

\usepackage{array}
\usepackage{graphicx}
\usepackage[caption=false,font=normalsize,labelfont=sf,textfont=sf]{subfig}
\usepackage{amsmath,amsfonts}
\usepackage{array}
\usepackage{pdfpages}
\usepackage{mathtools, cuted}
\usepackage{comment}
\usepackage{tcolorbox}

\newtcolorbox{myblock}{
  colback=blue!10,
  colframe=blue!50!black,
  fonttitle=\bfseries
}

\newcommand{\Rmnum}[1]

\usepackage{multirow}
\usepackage{amsthm,amssymb}
\newtheorem{theorem}{Theorem}[section]
\newtheorem{proposition}[theorem]{Proposition}
\newtheorem{lemma}[theorem]{Lemma}

\theoremstyle{definition}
\newtheorem{definition}[theorem]{Definition}

\newtheorem{mechanism}{Mechanism}
\newtheorem{remark}[theorem]{Remark}

\usepackage{bm}
\usepackage{graphicx}
\usepackage{makecell}
\usepackage{booktabs}

\usepackage{url}
\usepackage{textcomp}
\usepackage{stfloats}
\usepackage{url}
\usepackage{verbatim}
\usepackage{graphicx}
\usepackage{cite}
\begin{document}

\title{Incentivizing Social Information Sharing Through Routing Games}
\author{Songhua Li and Lingjie Duan, \IEEEmembership{Senior Member, IEEE}
\thanks{Songhua Li and Lingjie Duan are with the Pillar of Engineering Systems and Design, Singapore University of Technology and Design, Singapore.  \\Email: \{songhua\_li,lingjie\_duan\}@sutd.edu.sg.}
}

\markboth{
}%
{Songhua Li and Lingjie Duan: }


\maketitle

\begin{abstract}
Mobile crowdsourcing platforms leverage a mass of mobile users and their vehicles to learn massive point-of-interest (PoI) information while traveling and share it as a public good.  Given that the crowdsourced users mind their travel costs and have diverse usage preferences for PoI information along different paths, we formulate the problem as a novel non-atomic multi-path routing game. This game features positive network externalities from social information sharing, distinguishing it from the traditional routing game literature that primarily focuses on congestion control of negative externalities among users. Our price of anarchy (PoA) analysis shows
that in the absence of any incentive design, users' selfish routing on the lowest-cost path will significantly limit PoI diversity leading to an arbitrarily large efficiency loss from the social optimum with a PoA of $0$. This motivates us to design effective incentive mechanisms to remedy while upholding desirable properties including individual rationality (IR), incentive compatibility (IC), and budget balance (BB) to ensure practical feasibility.  Without knowing a specific user's path preference, we first present a non-monetary mechanism called Adaptive Information Restriction (AIR) that satisfies all the desirable properties. Our AIR indirectly penalizes non-cooperative users by adaptively reducing their access to the public good, according to actual user flows along different paths. It achieves a significant PoA of $\frac{1}{4}$ 
with low complexity $O(k\log k+\log m)$, where $k$ and $m$ represent the numbers of paths and user preference types, respectively. When the system can support pricing/billing for users, we further propose a new monetary mechanism called Adaptive Side-Payment (ASP), which adaptively charges and rewards users based on their chosen paths. 
Our ASP achieves a better PoA of $\frac{1}{2}$ with an even lower complexity of $O(k\log k)$. Finally, our theoretical findings are well corroborated by our experimental results using a real-world dataset. 
\end{abstract}
%


\maketitle
\section{Introduction}
\IEEEPARstart{T}{he} success of many location-based services is largely driven by the availability of abundant point-of-interest (PoI) information \cite{chang2021mobility,vasserman2015implementing,cheng2023influence}. Leveraging ubiquitous mobile social platforms \cite{nikou2014ubiquitous,hossain2020ubiquitous} and seamless Vehicle-to-Everything (V2X) networks \cite{garcia2021tutorial,chen2017vehicle}, 
one promising approach for gathering PoI data is to crowdsource mobile users and vehicles as they collect PoIs during their regular travels, without requiring much additional effort. For example, Waze and Google Maps encourage drivers to report traffic conditions along their routes, enriching the data pool for live mapping \cite{vasserman2015implementing}. TripAdvisor and Yelp rely on user contributions to gather shopping promotions and ratings for restaurants and hotels \cite{Tripadvisor, yelp}. These crowdsourcing platforms typically aggregate PoI information from individual users and make it broadly accessible, thereby enhancing users' future travel and service experiences. 
\begin{figure}[t]
    \centering
    \includegraphics[width=8cm]{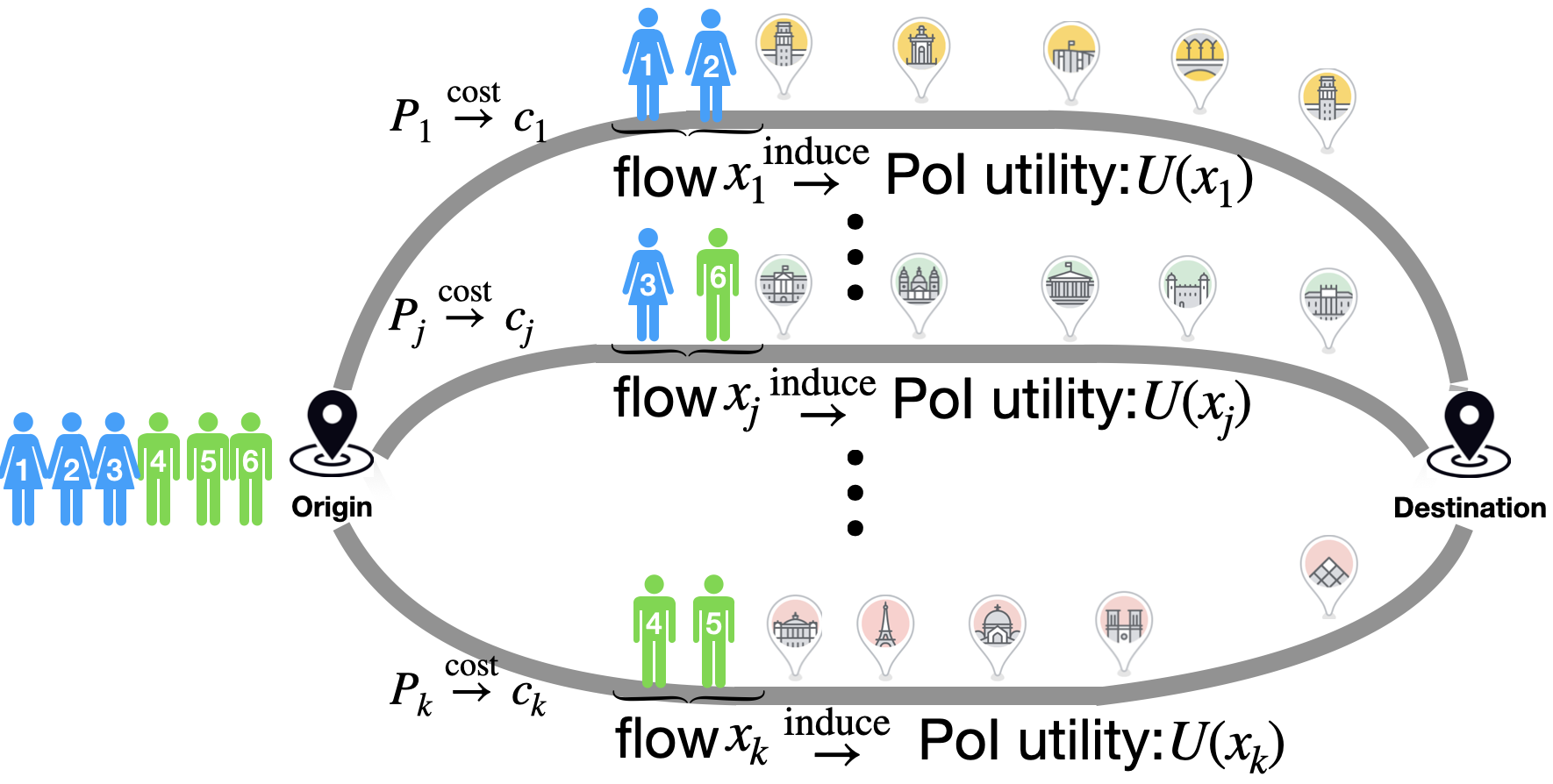}
    \caption{Network configuration in the routing game for social information sharing.
This illustrative example depicts a $k$-path routing game with $m=2$ user preference types, represented by blue and green. Users of the blue type prefer PoIs located along the upper paths, whereas users of the green type favor PoIs on the lower paths.}
    \label{fig01_multi}
\end{figure}

To motivate this study, let us consider a typical scenario of social information sharing in location-based services \cite{cominetti2022approximation,li2023congestion}: a large number of crowdsourced users select from multiple paths connecting a common origin and destination,  as illustrated in Fig.~\ref{fig01_multi} and further elaborated in the model description in Section~\ref{model_sec}. Each path contains abundant PoI information relevant to the public. However, an individual user can collect only a small subset of PoIs along her selected path, which is trivial compared to the massive PoIs collected collectively. These users, who may have heterogeneous preferences for PoIs located on different paths, are rational and tend to select routes that minimize their travel costs
\cite{
li2017dynamic}. This tendency inadvertently limits the diversity of the aggregated PoI information. Therefore, this paper aims to develop effective mechanisms to regulate user routing decisions, thereby enriching the diversity of collected PoIs and improving the overall social welfare derived from the 
PoI aggregation (i.e., the combined dataset of points of interest gathered from all users’ individual collections).

The scenario described above is closely related to routing games (e.g., \cite{macault2022social,cominetti2019price}) and algorithmic mechanism design (e.g., \cite{kim2021simple,li2018customer}). Since each user’s PoI contribution is infinitesimal compared to the PoI aggregation, we frame the problem as a non-atomic routing game, as similarly formulated in \cite{macault2022social}. 
%
Related works on routing and congestion games \cite{zhu2022information,macault2022social,li2024human,wu2021value,li2024optimize,giordano2023note,li2025analyze} focus on the equilibrium outcomes of users' selfish routing decisions, studying either informational mechanisms \cite{tavafoghi2017informational,farhadi2022dynamic} or pricing/monetary mechanisms \cite{kordonis2019mechanisms,kim2021simple,ghafelebashi2023congestion}. These studies typically address users' crowding interactions under congestion with negative network externalities \cite{tavafoghi2017informational,li2017dynamic}. In contrast, our routing game highlights the positive network externalities arising from social information sharing among crowdsourced users across multiple routing paths. Moreover, existing informational regulations in traditional routing or congestion games often involve systems with complete knowledge that intentionally conceal travel costs from users \cite{li2023congestion,mansour2022bayesian}. As a result, users cannot accurately determine the shortest path and tend to rely on system recommendations. Differently, our new routing game ensures transparent travel costs for users while preserving the privacy of users' path preferences from the system. This transparency allows users to easily select their preferred paths but significantly reduces PoI diversity, thereby making system intervention more challenging. Although a study \cite{li2017dynamic} explores social information sharing, it is confined to a simple two-path network and heavily relies on the assumption that the system possesses complete knowledge of each user's path preference. Hence, it is essential to design new mechanisms with guaranteed performance for the multi-path routing game that exhibits positive network externalities.

In addition to congestion and routing games, recent studies on cost-sharing games \cite{gkatzelis2016optimal, christodoulou2024resource} investigate scenarios 
where selfish users choose from a set of resources (e.g., paths or machines). In these games, the cost of each resource depends on the number of users selecting that resource and is subsequently shared among them \cite{gkatzelis2016optimal,georgoulaki2021equilibrium}. The policy-maker in cost-sharing games often focuses on determining the proportion of the resource cost that each user is responsible for \cite{gkatzelis2021resource,christodoulou2020resource}. In contrast, our routing game aims to optimize the flow distribution of users across multiple paths (i.e., resources) to maximize social utility. Since social utility in our game depends on the distribution of different types of users across various paths, transforming our game into a cost-sharing game is not feasible.

In algorithmic game theory and mechanism design, the price of anarchy (PoA) \cite{papadimitriou2001algorithms} is commonly used as a standard metric for evaluating the performances of mechanisms (e.g., \cite{li2017dynamic,christodoulou2024resource}). In the context of our game for social information sharing, the PoA of a mechanism ranges from 0 (indicating 
zero efficiency) to 1 (indicating optimality). The formal definition of the PoA used in this paper is provided later in Definition~\ref{def_poa} in Section~\ref{model_sec}.

\textit{Our main results are summarized below and in Table~\ref{mainresults_table}.} 

\begin{table}[t]
\caption{Main results of this paper, with $k$ indicating the number of paths and $m$ denoting the number of users' preference types 
}
\centering
\begin{tabular}{c|c|c|c}
\hline
 &\makecell[c]{No Mech.}&  \makecell[c]{AIR} &\makecell[c]{ASP}\\
\hline\hline
\makecell[c]{Efficiency}&$PoA=0$&$PoA=\frac{1}{4}$&${PoA=\frac{1}{2}}$\\
\hline
Complexity&{\textit{N.A.}} &
$O(k\log k+\log m)$&$O(k\log k)$\\
\hline
Property&\textit{N.A.}&IR, IC, and BB&IR, IC, and BB\\
\hline
Location&Lemma \ref{thm_no_incentive_basic} &Theorem~\ref{air_multipath_poa}&Theorem \ref{new_poa_for_spm}\\
\hline
\end{tabular}
\label{mainresults_table}
\end{table}
To begin with, we demonstrate through PoA analysis that in the absence of any mechanism design, users' selfish routing choices can lead to remarkably imbalanced flow distribution on the paths, resulting in very poor social information sharing with $PoA=0$. To remedy such a huge efficiency loss, we are motivated to design effective incentive mechanisms to regulate users' routing behaviors.  To ensure practical feasibility, all of our mechanisms adhere to key properties, including 
individual rationality (IR), incentive compatibility (IC), and budget balance (BB).

Our first mechanism, named adaptive information restriction (AIR), limits users' access to the PoI aggregation by imposing certain fractions as an indirect penalty.
By meticulously adapting its penalty fractions to users' actual flow on different paths, our AIR mechanism successfully achieves a significant PoA of
$\frac{1}{4}$ in polynomial time
$O(k\log k+\log m)$, which is for the routing game involving $k$ paths and $m$ types of users' path preferences.  To establish this PoA result, we recursively decompose the outer multi-path routing game into nested sub-routing games until reaching the base case of a two-path routing game, while ensuring that the equilibrium flow in each sub-game consistently aligns with that of its outer routing game. We believe this novel analysis approach is of independent interest for broader mechanism design and PoA analysis. This analysis technique can be of independent interest for broader mechanism design and analysis.

Our second mechanism, Adaptive Side Payment (ASP), is a monetary mechanism designed for systems that support user payment transactions. ASP employs a two-tier side-payment structure that adaptively charges and rewards users according to their path selections. In its first tier, side payments are carefully crafted and imposed on each path to calibrate its effective cost, ensuring alignment between these costs and users’ path preferences.
Its second tier further casts the multi-path game into a fundamental two-path game, using two pivotal thresholds determined by our proposed mathematical program. Leveraging its two-tier side-payment design, our ASP achieves an improved Price of Anarchy (PoA) of $\frac{1}{2}$, along with a reduced computational complexity of $O(k \log k)$. Finally, we validate our theoretical results through experiments on a real-world dataset. 

The remainder of this paper is organized as follows: Section~\ref{model_sec} describes our routing game for social information sharing and presents preliminary results. Our mechanism designs and PoA analyses for AIR and ASP mechanisms are presented in Sections~\ref{sec_IRM} and~\ref{sec_spm}, respectively. 
In Section~\ref{sec_simulation}, we present experimental results using a real-world dataset to validate the practical performance of our mechanisms.
Finally, Section~\ref{sec_conclusion} concludes this work. 
%
%
%
\section{The Model and Preliminaries}\label{model_sec}
We consider a non-atomic routing game with positive network externalities arising from social PoI sharing in mobile crowdsourcing platforms,  as depicted in Fig.~\ref{fig01_multi}. In the game, a normalized unit mass of rational users independently chooses one of the $k \geq 2$ parallel paths, $\mathcal{P} = {P_1, \ldots, P_k}$, to travel from an origin $O$ to a destination $D$ based on their individual preferences.\footnote{We reasonably assume that everyone in the game opts to travel from $O$ to $D$, which can be easily ensured by offering base incentives to encourage participation.} Users collect PoIs along their chosen paths to share as a public good for enhancing social welfare. 

Concretely, each path $P_j\in \mathcal{P}$ is associated with a publicly known\footnote{Travel costs in the game, including distance and fuel consumption, can be accurately predicted by modern navigation apps such as Waze and Google Maps \cite{derrow2021eta}.} travel cost $c_j$, where the smallest one is normalized to zero, i.e.,  $\min\limits_{j\in\{1,...,k\}}\{c_j\}=0$.
Denote $U(x)$ as the function quantifying the amount of distinct PoIs (e.g., landmarks, shopping promotions, restaurant ratings, traffic conditions) collected by a flow $x\in [0,1]$ of users on the same path. Due to overlap in users’ PoI collections, the number of distinct PoIs gathered from a given path generally increases as more users follow that path, but at a decreasing rate. Hence, function $U(x)$ is assumed to be monotonically increasing and concave in $x$. Our subsequent mechanisms and their PoA guarantees hold for any function $U(x)$ that is concave and monotonically increasing. A concrete form of $U(x)$ for implementation will be discussed later in Section~\ref{sec_simulation}. In the sequel, we use $a\rightarrow b$ to denote that the value of $a$ converges to $b$.

Since the additional costs users incur from route deviations are typically negligible compared to the long-term benefits of the PoI aggregation from all users' individual collections), we assume
$\frac{c_j}{U(1)}\rightarrow 0$ for each $P_j\in \mathcal{P}$.  As we will show later, this condition $\frac{c_1}{U(1)}\rightarrow 0$ induces the worst-case efficiency loss for the system, presenting a challenging scenario for mechanism design to remedy a severely degraded PoA.

To account for the varying significance of PoI information across paths, users are supposed to have heterogeneous preferences over the PoIs on different paths, which remain unknown to the mechanism designer to preserve user privacy \footnote{Nevertheless, the mechanisms we propose later do not rely on any knowledge of these preferences.}. Accordingly, we model the user population as comprising $m$ distinct types, indexed by $[m] \triangleq \{1, \ldots, m\}$, each defined by a unique preference profile. Let $\eta_i$ denote the proportion of users of type $i$, thereby satisfying
$\sum_{i=1}^m \eta_i = 1$. Users of the same type $i\in [m]$ share a common preference profile $\bm{W}_i=(w_{i1},...,w_{ik})$, where $w_{ij}$ gives a type-$i$ user's preference weight over PoIs on path $P_j \in \mathcal{P}$ and satisfy $\sum_{j=1}^k w_{ij}=1$. We compile the preference profiles of the $m$ user types into the set $\mathcal{W} = \{{\bm W}_1,...,{\bm W}_m\}$. Table~2 summarizes the key notation used throughout the paper.
\begin{table}[t]
\centering
    \caption{Key notation summary}
\begin{tabular}{c|p{5cm}}
    \hline
   Key notation & Description\\
    \hline\hline
 $\mathcal{P}\triangleq\{P_1, ..., P_k\}$& Set of $k$ parallel paths in the network\\
 \hline
 $c_j$ & Travel cost for path $P_j\in\mathcal{P}$\\
 \hline
 $[m]\triangleq\{1,...,m\}$& Set of user types\\
 \hline
 $\eta_i$& The proportion of users of type $i\in[m]$ within the total user population\\
 \hline
$\bm{W}_i=(w_{i1},...,w_{ik})$ & preference profile of type-$i$ users over $k$ paths; $w_{ij}$ is the preference weight for path $j$, with $\sum_{j=1}^k w_{ij} = 1$\\
\hline
$(x_1,\ldots,x_k)$ & User flow vector across the $k$ paths with $x_j$ giving user flow on path $j$ and $\sum\limits_{i=1}^k x_i=1$\\
\hline
$U(x)$ &PoI amount collected by a user flow $x$\\
\hline
$\mathcal{U}_i(x_1,...,x_k)$ in (\ref{utility_multi_path}) & Utility of type-$i$ users given user flows  $(x_1,\ldots,x_k)$\\
 \hline
 $\Theta_{i}(P_j,x_1,...,x_k)$ in (\ref{payoff_multi_path})& Payoff of type-$i$ users on path $j$ given user flows $(x_1, \ldots, x_k)$\\
 \hline
 $SW(x_1,...,x_k)$ in (\ref{sw_multi_path})& Social welfare resulting from flow vector $(x_1, \ldots, x_k)$\\
 \hline
    \end{tabular}

    \label{notationinouralgorithm}
\end{table}

Given the cost vector $\bm{c}=\{c_1,...,c_k\}$ and the function $U(x)$, each user in the game selfishly selects the path that maximizes her payoff, which is defined as follows. 
Let $\mathbf{x} = (x_1, \ldots, x_k)$ denote the flow vector of users, where $x_j \geq 0$ represents the fraction of users choosing path $P_j$, with $\sum{j=1}^k x_j = 1$. The utility for a type-$i$ user, with full access to the PoI aggregation collected from all users, is given by:
\begin{equation}\label{utility_multi_path}
\mathcal{U}_i(x_1,...,x_k) =\sum\limits_{j=1}^{k} w_{ij}U(x_j).
\end{equation}
The payoff for a type-$i$ user selecting path $P_j$, denoted as $\Theta_{i}(P_j,x_1,...,x_k)$, is defined as her utility from the PoI aggregation minus her travel cost as given below.
\begin{equation}\label{payoff_multi_path}
   \Theta_{i}(P_j,x_1,...,x_k) =\mathcal{U}_i(x_1,...,x_k)-c_j.
\end{equation}
Thus, for the normalized unit mass of users, the \textit{social welfare} $SW(x_1, \ldots, x_k)$ corresponding to the flow vector $(x_1, \ldots, x_k)$ is defined as the total sum of all users’ payoffs:
\begin{equation}\label{sw_multi_path} 
\begin{split}
SW(x_1,...,x_k)
&=\sum\limits_{i=1}^m \eta_i \mathcal{U}_i(x_1,...,x_k)-\sum\limits_{j=1}^kx_jc_j.
\end{split}
\end{equation}
By substituting (\ref{utility_multi_path}) into (\ref{sw_multi_path}), the social welfare can be expressed as in (\ref{sw_multi_path_new}). Notably, the coefficient $\sum_{i=1}^m \eta_i w_{ij}$ of $U(x_j)$ in the first summation term can be regarded as the users’ expected preference for path $P_j$. This insight plays a crucial role in guiding our subsequent mechanism designs. Accordingly, we highlight $\sum_{i=1}^m \eta_i w_{ij}$ and refer to it as the \textit{social utility weight} associated with $U(x_j)$.
\begin{equation}\label{sw_multi_path_new}
SW(x_1,...,x_k)
=\sum\limits_{j=1}^k(\sum\limits_{i=1}^m\eta_i w_{ij})U(x_j)-\sum\limits_{j=1}^kx_jc_j.
\end{equation}
\begin{remark}\label{remark_01}
Beyond PoI-based mobile crowdsourcing, the above multi-path routing game may accommodate broader application scenarios by treating each path as a prescribed category of tasks that the platform aims for users to complete to enhance social welfare and considering each PoI along a path as an individual task within that category. For example, Google Crowdsource \cite{Googlecrowdsourcea} allows users to select from various task categories (such as image labeling, translation validation, handwriting recognition, and map transcription) and complete micro-tasks within their chosen category. The collected data then serves as input for training models that enhance Google’s AI and language tools as a public good.
\end{remark}

In this paper, we focus on the equilibrium flow of users, as formally defined below and commonly studied in the literature \cite{roughgarden2002bad,richman2007topological,altman2001routing,li2017dynamic}.
\begin{definition}[Equilibrium flow]\label{NE_def}
Given the input parameters $(\bm{\eta}, \mathcal{W}, {\bm c}, U(\cdot))$ of the routing game, a feasible flow $(x_1,...,{x}_k)$ is called an \textit{equilibrium flow} if no user of any type $i\in[m]$ selecting path $P_j\in \mathcal{P}$ (i.e., in flow  $x_j$) can increase her payoff in (\ref{payoff_multi_path}) by unilaterally switching to a different path. We denote the set of all equilibrium flows under the given parameters as $E(\bm{\eta}, \mathcal{W}, \bm{c}, U(\cdot))$, and simply write $E$ when the context is clear. We use a hat to indicate an equilibrium flow, i.e., $(\widehat{x}_1, \ldots, \widehat{x}_k) \in E$.
\end{definition}
To evaluate the efficiency of mechanisms in the routing game, we adopt the Price of Anarchy (PoA) \cite{papadimitriou2001algorithms} as formally defined below.
\begin{definition}\label{def_poa}
The Price of Anarchy (PoA) refers to the ratio between the lowest social welfare at the worst-case equilibrium and the optimal social welfare, i.e., 
\begin{equation}\label{form_poa}
    {  PoA} \triangleq \min\limits_{(\widehat{x}_1,...,\widehat{x}_k)\in E(\bm{\eta}, \mathcal{W}, {\bm c}, U(\cdot))} \frac{SW(\widehat{x}_1,...,\widehat{x}_k)}{SW^*},
\end{equation}
where $SW^*\triangleq \max\limits_{0\leq x_1,...,x_k\leq 1, \sum_{i=1}^k x_i=1} \left\{SW(x_1,...,x_k)\right\}$.
\end{definition}
Note that the PoA in our routing game ranges from 0 to 1, where a PoA of 1 indicates optimal social welfare and a PoA of 0 signifies complete inefficiency. Without any incentive mechanism, we show in the following lemma that users' selfish path selections can lead to an extremely poor PoA.
\begin{lemma}\label{thm_no_incentive_basic}
In the routing game for social information sharing, the PoA collapses to zero when no mechanism is implemented to regulate users’ selfish path selections.
\end{lemma}

This negative result highlights the need for effective mechanisms to address the substantial efficiency loss in social PoI sharing. To ensure practical viability, we particularly focus on mechanisms that satisfy the following desirable properties.

\begin{definition}\label{def_property}
We consider the following three desirable properties, commonly studied in the mechanism design literature \cite{bergemann2005robust,li2017dynamic,balseiro2019dynamic}:
\begin{itemize}
    \item \textbf{Individual Rationality (IR)}: Each user receives a non-negative payoff from participating in the mechanism.
    \item \textbf{Incentive Compatibility (IC)}: Each user of type $i \in [m]$ is incentivized to choose a path that aligns with her true preference profile $\bm{W}_i$.
    \item \textbf{Budget Balance (BB)}: Any charges (if imposed) collected from users are fully redistributed among users within the system.
\end{itemize}
\end{definition}
%

In the sequel,  we begin by examining non-monetary mechanisms in Section~\ref{sec_IRM}, emphasizing their straightforward implementation. Then, in Section~\ref{sec_spm}, we address monetary/pricing mechanisms suited for systems when user transactions are supported. Through our subsequent PoA analysis, we demonstrate that appropriate pricing can further improve social welfare compared to non-monetary mechanisms.
%
%
%
\section{New Non-Monetary Mechanism: Adaptive Information Restriction}\label{sec_IRM}
In this section, we present an Adaptive Information Restriction (AIR) mechanism, as outlined in Mechanism~\ref{m_irm}, which regulates users’ access to the PoI aggregation via a set of parameters $\gamma$, referred to as \textit{penalty fractions}. At the high level, these fine-grained fractions dynamically control the extent to which users selecting certain paths can access the PoI aggregation, thereby serving as an indirect penalty to regulate users' path selections for better social welfare. Our subsequent Algorithms~\ref{incomplete_air_mechanism} and~\ref{general_incomplete_air_mechanism} tailor the penalty fractions accordingly.
%

\begin{mechanism}[\textbf{Adaptive Information Restriction (AIR)}]\label{m_irm}
Under the AIR mechanism, any user choosing path $P_j$ will eventually receive only a proportion $\gamma_j$ of the overall PoI aggregation described in (\ref{utility_multi_path}), as the remaining proportion $(1-\gamma_j)$ is excluded from her access as an indirect penalty.
\end{mechanism}
\begin{remark}
In practice, the penalty fractions in AIR can also be interpreted as follows: the smallest penalty fraction corresponds to the baseline service level provided uniformly to all users. The difference between any given penalty fraction and this baseline reflects the additional service benefits awarded to users who select the associated path. In other words, 
users will receive upgraded service when selecting paths with higher penalty fractions. Accordingly, our AIR mechanism can be implemented on platforms such as Yelp, Google Maps, and Waze, which offer higher service levels or rights to users who contribute PoIs in specific locations.
For example, Yelp recognizes top contributors as Elite members and offers them enhanced services, including more personalized recommendations \cite{dai2018aggregation,yelp}. By contributing PoIs along specific routes or in designated regions, users in Google Maps \cite{googlemapspoints,Googlecrowdsourcea} can accumulate points that increase their contributor level, thereby unlocking enhanced services. In Waze’s map editing system~\cite{wazepoints}, users who contribute PoIs to specific roads and places can earn more points and higher rank, which in turn grants them increased access to Waze’s data pool through broader editing rights. Therefore, our penalty fraction design in AIR can offer insights to these platforms on how to determine appropriate service levels that effectively incentivize users to contribute the specific PoI data the platforms seek.
Moreover, in the broader application context of the routing game discussed in Remark~\ref{remark_01}, our AIR mechanism can support data collection for large-scale model training \cite{zhang2024survey} by granting upgraded services to users who contribute data in specific domains or task categories prioritized by the platform.
\end{remark}

Given a flow vector $(x_1,\ldots, x_k)$, the payoff of a type-$i \in [m]$ user selecting path $P_j \in \mathcal{P}$ is reduced by the AIR mechanism from (\ref{payoff_multi_path}) to the following:
\begin{equation} \label{new_payoff_air}\Theta_i(P_j,x_1,...,x_k)=\gamma_j\mathcal{U}_i(x_1,...,x_k)-c_j.
\end{equation}

We now detail our design for penalty fractions in the AIR mechanism, aiming to ensure that the resulting equilibrium flows achieve high social welfare. Motivated by Lemma~\ref{thm_no_incentive_basic} and the insight that balanced flows across all paths are more resilient to worst-case scenarios, our design goal is to promote such balance through the choice of penalty fractions. To begin with, we consider the design in a fundamental two-path routing game, which lays the groundwork for the design in more complex multi-path scenarios.
%
\subsection{Warm-up: Adaptive Access Fraction Design for Two-path Routing Games}\label{sec_warmup_airr}
We first design for a simplified two-path routing game that involves only a single user type, whose preference weights are $w_0\leq 1$ and ($1-w_0$) for the PoI collections from paths $P_1$ and $P_2$, respectively. Here, path $P_1$ has a high cost $c_1$, while path $P_2$ has a low cost normalized to zero ($c_2=0$). 

Our design is presented in (\ref{fractions_rsm_homo}), where the high-level idea is to gradually grant greater access to users who follow higher-cost paths with lower flows. Concretely, to encourage balanced flows for when $w_0< 1$, our penalty fractions for AIR are dynamically adjusted by the first two cases of (\ref{fractions_rsm_homo}), incentivizing users to gradually shift towards the path they value less.
In the special case where users derive value exclusively from the PoI collection on the high-cost path (i.e., $w_0 = 1$), we assign this path a slightly higher penalty fraction than the low-cost path, see the last two cases of (\ref{fractions_rsm_homo}). This nuanced adjustment, proportional to the cost difference $|c_1-c_2|$ relative to the total PoI collection $U(1)$, is instrumental in ensuring a favorable PoA.
\begin{equation}\label{fractions_rsm_homo}
\gamma_{j}(x_1,x_2)=\left\{\begin{matrix}
 \frac{x_2}{x_1+x_2},& {\rm if}\;w_0<1{\rm\;and\;}j=1,\\ 
\frac{x_1}{x_1+x_2}, &{\rm if}\;w_0<1{\rm\;and\;}j=2, \\ 
 \frac{1}{2}+\frac{c_1}{2U(1)},&{\rm if}\;w_0=1{\rm\;and\;}j=1, \\ 
\frac{1}{2}-\frac{c_1}{2U(1)}, &{\rm if}\;w_0=1{\rm\;and\;}j=2. 
\end{matrix}\right.
\end{equation}

We proceed to examine the PoA under our AIR mechanism with penalty fractions given in (\ref{fractions_rsm_homo}). In the two-path routing game with a single user type, any equilibrium flow $(\widehat{x}_1,\widehat{x}_2)$ renders users indifferent between the two paths, i.e.,
\begin{equation}\label{ef_homogeneous}
\Theta(P_1,\widehat{x}_1,\widehat{x}_2)=\Theta(P_2,\widehat{x}_1,\widehat{x}_2).
\end{equation}
By substituting (\ref{fractions_rsm_homo}) into (\ref{new_payoff_air}) and then into (\ref{ef_homogeneous}), we establish in Lemma~\ref{theorem_irm_homo_poa} that the equilibrium flow in AIR converges to a balanced distribution of $(\frac{1}{2},\frac{1}{2})$ across the two paths, resulting in a PoA of $\frac{1}{4}$. This result provides a foundational step toward generalizing our design of (\ref{fractions_rsm_homo}) to multi-path routing games as we will show later in Section~\ref{section_matryosha}.

\begin{lemma}\label{theorem_irm_homo_poa}
In a two-path routing game with a single user type, the AIR mechanism using penalty fractions in (\ref{fractions_rsm_homo}) satisfies IR, IC, and BB, and achieves a PoA of $\frac{1}{4}$ relative to the social optimum. 
\end{lemma}
%

We further investigate a two-path routing game with two distinct user types.  Type~1 users, constituting a fraction $\eta$ of population, and type-2 users, making up the remaining proportion $1-\eta$, have preference profiles $(w_1,1-w_1)$ and $(w_2,1-w_2)$ over the PoI collections on paths $(P_1, P_2)$, respectively. \textit{W.l.o.g.}, suppose $\eta > 1 - \eta$, i.e., $\eta > \frac{1}{2}$. Based on (\ref{ef_homogeneous}), any equilibrium flow $(\widehat{x}_1,1-\widehat{x}_1)$ under AIR, which makes type-$1$ users indifferent between the two paths, must satisfy: 
\begin{equation}\label{condition_for_incomplete_air}
    (1-2\widehat{x})[w_1U(\widehat{x})+(1-w_1)U(1-\widehat{x})]=c_1.
\end{equation}
To analyze the solution $\widehat{x}$ to (\ref{condition_for_incomplete_air}) above, we define a new function $\mathcal{G}(x)$ as the left-hand side expression of (\ref{condition_for_incomplete_air}):
\begin{equation}\label{def_mathcalG_incomplete}
\begin{split}
\mathcal{G}(x)\triangleq(1-2x)\cdot [w_1U(x)+(1-w_1)U(1-x)]
\end{split}
\end{equation}
Since both $1-2x$ and $\mathcal{U}_1(x,1-x)= w_1U(x)+(1-w_1)U(1-x)$ in (\ref{utility_multi_path}) are non-negative and concave for $x\in[0,\frac{1}{2}]$, we can conclude the following.

\begin{lemma}\label{lemma_incomplete_air_01}
$\mathcal{G}(x)$ is concave and continuous in  $x\in[0,\frac{1}{2}]$.
\end{lemma}
\begin{figure}[t]
    \centering
\includegraphics[width=7cm]{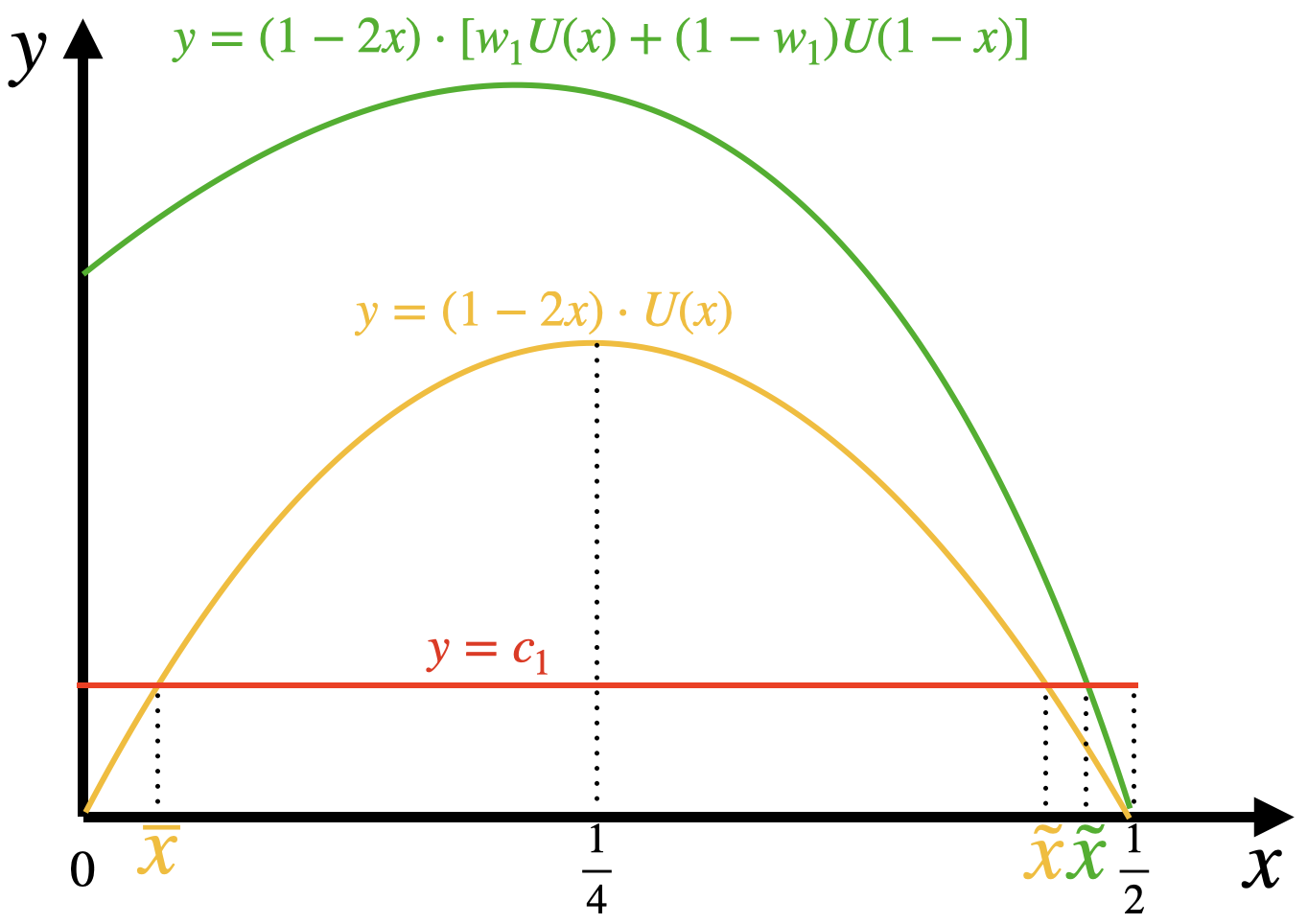}
    \caption{Illustration of solutions to $\mathcal{G}(x)=c_1$, in which the yellow curve represent the special case when $w_1=1$.}
    \label{fig04}
\end{figure}

In Fig.~\ref{fig04}, the yellow and the green curves represent the function $\mathcal{G}(x)$ for the cases $w_1=1$ and $w_1<1$, respectively. We observe that the curve $y=\mathcal{G}(x)$ always intersects the red line $y=c_1$ at some point $\widetilde{x}\in(\frac{1}{4},\frac{1}{2})$. Since $c_1$ is significantly smaller than $U(1)$ due to $\frac{c_1}{U(1)}\rightarrow 0$ as assumed in Section~\ref{model_sec}, the curve of $y=\mathcal{G}(x)$ also intersects the red line $y=c_1$ at another point $\overline{x}\in 
(0,\frac{1}{4})$. The second intersection occurs only when $w_1$ is sufficiently close to 1, such that $c_1 > (1 - w_1)U(1)$. These findings are summarized in the following lemma.
\begin{lemma}\label{lemma_incomplete_air_02}
If $w_1\geq 1-\frac{c_1}{U(1)}$, the equation in (\ref{condition_for_incomplete_air}) admits two solutions $\overline{x}\rightarrow0^+$ and $\widetilde{x}\rightarrow \frac{1}{2}^-$; otherwise, the equation in (\ref{condition_for_incomplete_air})  admits a unique solution $\widetilde{x}\rightarrow \frac{1}{2}^-$.  
\end{lemma}
\begin{algorithm}[t]
\caption{Adaptive fractions for AIR in two-path}\label{incomplete_air_mechanism}
   \SetKwInOut{Input}{Input}
        \SetKwInOut{Output}{Output}
  \Input{$w_1,w_2,(\eta,1-\eta), c_1, U(x)$.}
  \Output{penalty fractions $\gamma_{j}$.}
\eIf{(\ref{condition_for_incomplete_air}) admits a solution $\overline{x}\in[0,0.25]$ such that $1-\eta\leq \overline{x}$}
{
Apply penalty fractions in (\ref{fractions_rsm_homo_incomplete_02}).
}
{
Apply penalty fractions in (\ref{fractions_rsm_homo_incomplete_01}).
}
\end{algorithm}
Besides the desired balanced flow $\widetilde{x}$, Lemma \ref{lemma_incomplete_air_02} reveals the possibility of an imbalanced flow $\overline{x}$ that may lead to poor social welfare when type-1 users dominate the population in the game, i.e., when $\eta\geq 1-\overline{x}\geq 0.75$. To prevent this undesirable outcome, we introduce the following penalty fractions in (\ref{fractions_rsm_homo_incomplete_02}) to steer the equilibrium flow toward balance: 
\begin{equation}\label{fractions_rsm_homo_incomplete_02}
    \gamma_{j}=\left\{\begin{matrix}
 \frac{1}{2}+\frac{c_1}{2U(1)},&{\rm when}\;j=1, \\ 
\frac{1}{2}-\frac{c_1}{2U(1)}, &{\rm when}\;j=2. 
\end{matrix}\right.
\end{equation}
When $\overline{x}$ does not exist, we simply adopt the following fraction design derived from (\ref{fractions_rsm_homo}):
\begin{equation}\label{fractions_rsm_homo_incomplete_01}
    \gamma_{j}=\left\{\begin{matrix}
 \frac{x_2}{x_1+x_2},& {\rm when\;}j=1,\\ 
\frac{x_1}{x_1+x_2}, &{\rm when}\;j=2.
\end{matrix}\right.
\end{equation}

Algorithm~\ref{incomplete_air_mechanism}summarizes our fraction design for a two-path routing game involving two user types, enabling our AIR mechanism to achieve favorable social welfare while maintaining desirable properties, as shown in the proposition below. 
\begin{proposition}\label{theorme_air_incomplete}
For a two-path routing game with two user types, the 
AIR mechanism with penalty fractions given by Algorithm \ref{incomplete_air_mechanism} satisfies IC, IR, and BB, and achieves $PoA=\frac{1}{4}$.
\end{proposition}
Next, we extend the design of AIR's penalty fractions to a general multi-path routing game with more user preference types. 
\subsection{Generalization: Matryoshka Fraction Design for Multi-path Routing Games}\label{section_matryosha}
%
%
To address general routing games with $k\geq 2$ paths and $m\geq 2$ user types, we propose a recursive decomposition approach that outcomes a sequence of sub-routing games, forming a nested structure reminiscent of Matryoshka dolls. Accordingly, we refer to this penalty fraction design as Matryoshka, as detailed in Algorithm~\ref{general_incomplete_air_mechanism}.

Specifically, in each recursive step, the multi-path routing game is reduced to an (outer) two-path game by partitioning the set of paths into two parts. Each part serves as an aggregate path in the outer game and, if it contains multiple paths, forms a sub-routing game to be recursively processed in the next step. Notably, 
the paths involved in an inner sub-routing game all originate from the same part of its corresponding outer routing game in the previous recursion step.


\begin{algorithm}[ht]
\caption{Matryoshka fraction design for AIR in multi-path}\label{general_incomplete_air_mechanism}
\SetKwInOut{Input}{Input}
\SetKwInOut{Output}{Output}

\Input{ $w_{ij}$, $\eta_i$, $x_j$ and $c_j$ that correspond to each type $i\in [m]$ of users and each path $P_j$ ($j\in [k]$).}
 Arrange paths $\{P_j|j\in[k]\}$ in descending order of their utility weights $[\sum\limits_{i=1}^m \eta_i w_{ij}]$, resulting in $P_1, P_2, ..., P_k$.\\
 Find the dominating user type $\mu\triangleq \arg\max\limits_{i\in[m]}\{\eta_i\}$ which accounts for the largest proportion $\eta_{\mu}$ among all users.\\
\ForAll{$\lambda = 1:k-1$}{
 According to path $P_{\lambda}$, partition paths $\{P_j|j=\lambda+1,...,k\}$ into two groups:
$L_{\lambda}\triangleq \{P_j\in \Psi_{\lambda}|c_j\leq c_{\lambda}\}$ and $H_{\lambda}\triangleq \{P_j\in \Psi_{\lambda}|c_j> c_{\lambda}\}$.\\
 Find in set $L_{\lambda}$ the path $P_{\psi}=\arg\min\limits_{P_j\in L_{\lambda}}{\{c_j\}}$ with the lowest travel cost $c_{\psi}$.\\
\eIf{(\ref{condition_for_incomplete_air_general}) admits a solution $\overline{x}\in[0,0.25]$ such that $1-\eta_{\mu}\leq \overline{x}$}{  
Apply penalty fractions in (\ref{matryoshka_fractions1}).\\
 \textbf{break}}
{
Apply penalty fractions in (\ref{matryoshka_fractions2}).
}
}
\Output{penalty restrictions $\gamma_{j} = \prod_{\lambda = 1}^{k-1}\gamma_{j\lambda}$ for the path $P_j$.}
\end{algorithm}
The effectiveness of Algorithm~\ref{general_incomplete_air_mechanism} hinges on maintaining consistent alignment between the equilibrium flow within each inner sub-routing game and that of its corresponding outer routing game. Below, we discuss how such alignment is achieved.

In each recursive step of Algorithm~\ref{general_incomplete_air_mechanism}, type-$\mu$ users are made indifferent between paths $P_{\lambda}$ and $P_{\psi}$ by ensuring the condition $\Theta_{\mu}(P_{\lambda},x,1-x) = \Theta_{\mu}(P_{\psi},x,1-x)$ holds, that is, 
\begin{equation}\label{condition_for_incomplete_air_general}
    (1-2{x})[\frac{w_{\mu,\lambda}}{w_{\mu,\lambda}+w_{\mu,\psi}}U({x})+\frac{w_{\mu,\psi}}{w_{\mu,\lambda}+w_{\mu,\psi}}U(1-{x})]=c_{\lambda}-c_{\psi}.
\end{equation}
Paths $P_{\lambda}$ and $P_{\psi}$ in the above (\ref{condition_for_incomplete_air_general}) are selected in Lines~4–5 of Algorithm\ref{general_incomplete_air_mechanism} from the same partition of the current two-path routing game, with their costs satisfying $c_{\lambda} \geq c_{\psi}$.

Based on whether Equation~(\ref{condition_for_incomplete_air_general}) admits a solution $\overline{x} \in [1 - \eta_{\mu}, 0.25]$, Algorithm~\ref{general_incomplete_air_mechanism} selects appropriate penalty fractions according to the following two cases.

\noindent\textbf{Case 1:} When (\ref{condition_for_incomplete_air_general}) admits a solution $\overline{x}\in[1-\eta_{\mu},0.25]$, that is, when the \textbf{if} condition in Line~6 of Algorithm~\ref{general_incomplete_air_mechanism} is satisfied, the penalty fraction assigned to path $P_j$ in the current recursive round $\lambda$ is given by:
\begin{equation}\label{matryoshka_fractions1}
\gamma_{j\lambda}=\left\{\begin{matrix}
1,&{\rm if}\;P_j\in \{P_1,...,P_{\lambda-1}\},\\
 \frac{1}{2}+\frac{c_{\lambda}-c_{\psi}}{2U(1)},&{\rm if}\;P_j\in\{P_{\lambda}\}\cup H_{\lambda}, \\ 
\frac{1}{2}-\frac{c_{\lambda}-c_{\psi}}{2U(1)}, &{\rm if}\;P_j\in L_{\lambda},
\end{matrix}\right.
\end{equation}
The last two lines of (\ref{matryoshka_fractions1}) can be interpreted as assigning penalty fractions for an inner two-path routing game. In this inner game, the upper part comprises the path $P_{\lambda}$ and the set $H_{\lambda}$ of paths that have higher costs than $P_{\lambda}$ (as identified in Line~4 of Algorithm~\ref{general_incomplete_air_mechanism}). The lower part in the inner game consists of paths in $L_{\lambda}$, all with lower costs than $P_{\lambda}$ due to the Line~4 of Algorithm~\ref{general_incomplete_air_mechanism}). 

Notably, the equilibrium flow in the upper and lower parts is concentrated on paths $P_{\lambda}$ and $P_{\psi}$ (as selected in Line~5 of Algorithm~\ref{general_incomplete_air_mechanism}), respectively, since these two paths have the lowest costs within their parts. As the recursion ends here, the equilibrium flow on each path remains consistent between the inner and outer routing games.

\noindent\textbf{Case 2:} When the \textbf{if} condition in Line~6 of Algorithm~\ref{general_incomplete_air_mechanism} is not satisfied,  the penalty fraction imposed on path $P_j$ in the current recursion round $\lambda$ follows:
\begin{equation}\label{matryoshka_fractions2}
\gamma_{j\gamma}=\left\{\begin{matrix}
1,& {\rm if}\;j\in[\lambda-1],\\
\sum\limits_{t=\lambda+1}^k x_t/\sum\limits_{t=\lambda}^k x_t,& {\rm if\;}j={\lambda},\\ 
x_{\lambda}/\sum\limits_{t=\lambda}^k x_t, &{\rm otherwise},
\end{matrix}\right.
\end{equation}
Similar to Case~1, the above (\ref{matryoshka_fractions2}) also treats the multi-path game as an inner two-path game. Here, the upper part consists of a single path $P_{\lambda}$ which has the highest utility weight among all paths involved in the current round, while the lower part includes the remaining paths of the current round. Analogous to (\ref{fractions_rsm_homo_incomplete_01}), the flow-dependent fractions in the last two lines of (\ref{matryoshka_fractions2}) give the penalty fractions assigned to the upper and lower parts of the current recursive step, respectively.

In Algorithm~\ref{general_incomplete_air_mechanism}, Line~1 (for path arrangement) ensures that the path with the highest utility weight among those in the same recursive step is always assigned to the upper part.  This characteristic allows us to extend Proposition~\ref{theorme_air_incomplete} to the following theorem.

\begin{theorem}\label{air_multipath_poa}
   For a general multi-path routing game with $k$ paths and $m$ user types, the AIR mechanism employing penalty fractions from Algorithm~\ref{general_incomplete_air_mechanism} satisfies IC, IR, and BB. It runs in $O(k \log k+\log m)$ time and guarantees $PoA =\frac{1}{4}$.
\end{theorem}
As established in Theorem~\ref{air_multipath_poa}, our AIR mechanism achieves the first constant PoA in polynomial time for multi-path routing games, substantially improving upon the equilibrium in Lemma~\ref{thm_no_incentive_basic}, which yields a PoA of zero and suffers arbitrarily large efficiency loss. 
While a gap remains in comparison to the social optimum, our AIR mechanism highlights its significantly low time complexity, given that finding a polynomial-time optimal equilibrium is typically intractable. It is also worth noting that, as demonstrated by our experiments on a real-world dataset in Section~\ref{sec_simulation}, the AIR mechanism performs empirically even closer to the social optimum than its theoretical PoA guarantee.

In the next section, we further explore pricing/monetary mechanisms for platforms that support users' payment transactions.
\section{New Monetary Mechanism: Adaptive Side-payment}\label{sec_spm}
We first consider a typical two-path game to inspire the design for general multi-path routing games.
\subsection{Warm-up: Two-path Routing Game}
The warm-up routing game involves two paths and two user types, adopting the notation from AIR to specify its setting, as described right after Lemma~\ref{theorem_irm_homo_poa} in Section~\ref{sec_warmup_airr}. 
To address this, we present an adaptive side-payment (ASP) mechanism as outlined below in Mechanism~\ref{two-path-m}.
%
\begin{mechanism}[\textbf{Adaptive Side-Payment Mechanism, ASP}]\label{two-path-m}
ASP charges each user on the low-cost path $P_2$ a fixed payment $\tau$, and redistributes the total collected payment  $x_2\tau$ evenly to users on the high-cost path $P_1$, ensuring budget balance. 
\end{mechanism}
The critical parameter $\tau$ in ASP directly influences the efficiency of ASP and will be optimized as follows. To facilitate this optimization, we first present the following Lemma~\ref{condition_NE_SPM}, which offers a crucial condition for identifying the optimal $\tau$ associated with an equilibrium flow $(\widehat{x}_1,\widehat{x}_2)$ under ASP.
%
\begin{lemma}\label{condition_NE_SPM}
Any equilibrium flow $(\widehat{x}_1,\widehat{x}_2)$ in ASP satisfies $\frac{\widehat{x}_2}{\widehat{x}_1}=\frac{c_1-\tau}{\tau}$, which is equivalent to
\begin{equation}\label{uniqueness_SPM_NE}
    (\widehat{x}_1,\widehat{x}_2)=(\frac{\tau}{c_1},1-\frac{\tau}{c_1}).
\end{equation}
\end{lemma}
To further characterize the optimal $\tau$, we give the following Lemma~\ref{lemma_spm_001}, which identifies a fixed point $(\frac{1}{2},U(\frac{1}{2}))$ of the utility function $\mathcal{U}_i$ in (\ref{utility_multi_path}), and its subsequent Lemma~\ref{monotonicity_property}, which
provides the monotonicity intervals of each $\mathcal{U}_i$.
\begin{lemma}\label{lemma_spm_001}
$\mathcal{U}_i(x,1-x)+\mathcal{U}_i(1-x,x)=U(x)+U(1-x)$ and  $\mathcal{U}_i(\frac{1}{2},\frac{1}{2})=U(\frac{1}{2})$.
\end{lemma}
\begin{lemma}\label{monotonicity_property}
If  $\frac{1}{2}<w_i\leq1$, $\mathcal{U}_i(x,1-x)$ monotonically increases as $x$ increases in $[0,\frac{1}{2}]$ . If  $0<w_i\leq\frac{1}{2}$, $\mathcal{U}_i(x,1-x)$ monotonically decreases as $x$ increases in $[\frac{1}{2},1]$.
\end{lemma}
 \begin{figure}[t]
     \centering
     \includegraphics[width=7cm]{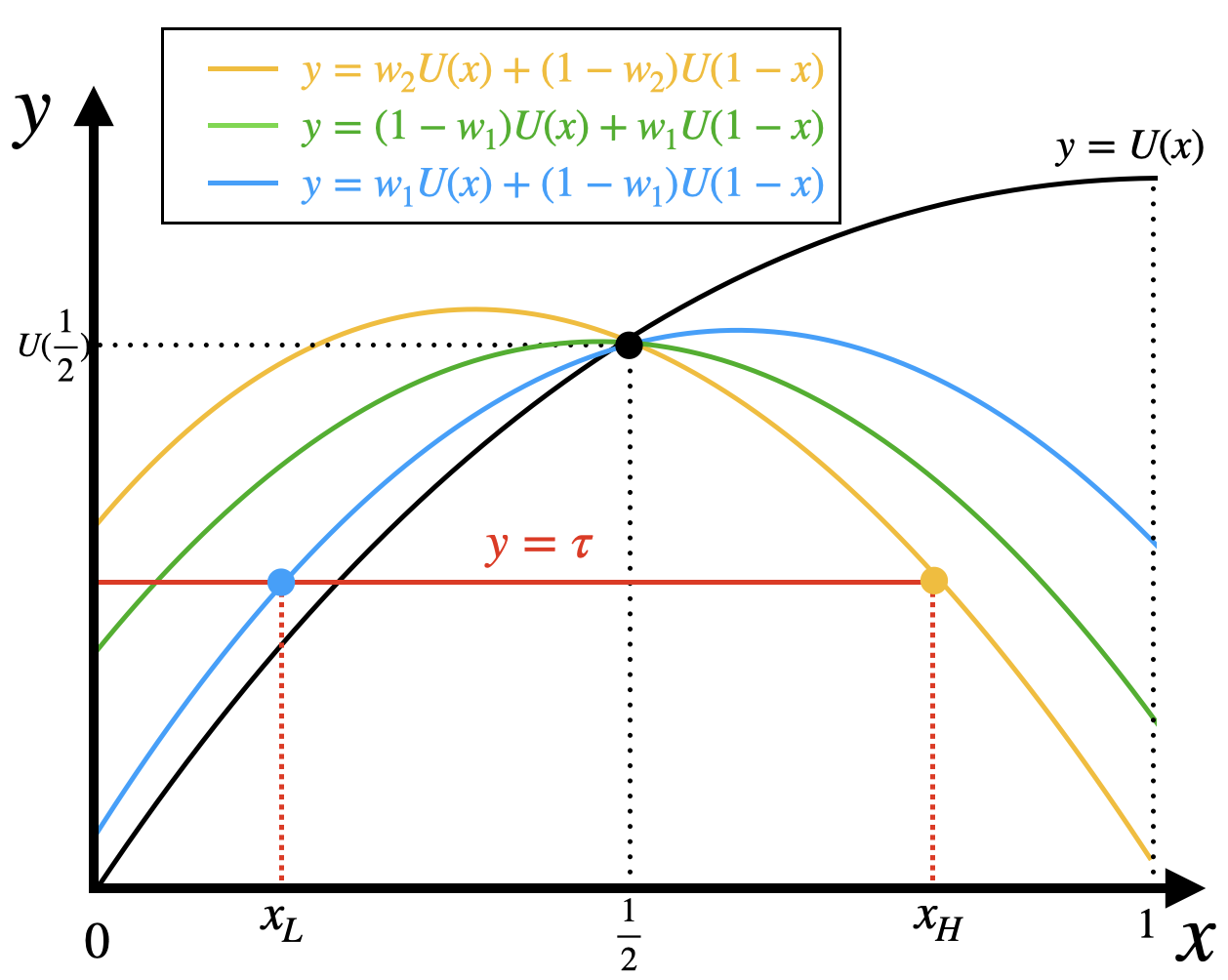}
     \caption{An illustration example for properties of concavely increasing function $U(x)$ and $\mathcal{U}_i(x,1-x)$ in (\ref{utility_multi_path}).}
     \label{fig02}
 \end{figure}
For the case $w_1>w_2$, Fig.~\ref{fig02} depicts the utility functions $\mathcal{U}_1$ and $\mathcal{U}_2$, defined in (\ref{utility_multi_path}), as the blue and yellow curves, respectively. In particular, when $w_1+w_2=1$, functions $\mathcal{U}_1$ (in blue) and $\mathcal{U}_2$ (in green) are symmetric about the vertical line $x=\frac{1}{2}$, revealing the fixed point $(\frac{1}{2},U(\frac{1}{2}))$ described in Lemma~\ref{lemma_spm_001}. As Lemma~\ref{monotonicity_property} indicates and Fig.~\ref{fig02} illustrates, the horizontal line $y=\tau$ intersects $\mathcal{U}_1$ and $\mathcal{U}_2$ at 
points $(x_L,\mathcal{U}_1(x_L,1-x_L))$ and $(x_H,\mathcal{U}_2(x_H,1-x_H))$, respectively, with $x_L<0.5<x_H$. This observation facilitates the following lemma,  which narrows down the feasible range for the optimal $\tau$.
\begin{lemma}\label{condition_both_ir_ic}
 When $\tau \leq U(\frac{1}{2})$, the ASP mechanism satisfies both IR and IC.
 \end{lemma}
Building on the above, we formulate Problem \RomanNumeralCaps{1} in (\ref{obj_best_tau_SPM})-(\ref{last_constraint_SPM}) to optimize $\tau$:
\begin{subequations}
\begin{align}
\mbox{\textbf{Problem \RomanNumeralCaps{1}:}}&\;\;\tau^* = \arg\max\limits_{\tau}\{ \alpha_1U(\frac{\tau}{c_1})+\alpha_2U(1-\frac{\tau}{c_1})-\tau\} &\label{obj_best_tau_SPM}\\
\mbox{\textbf{s.t.}} & \; \;c_1 x_L \leq \tau \leq c_1 x_H,&\label{constraint01_SPM}\\
& \; \;\mathcal{U}_1(x_L,1-x_L) =\mathcal{U}_2(x_H,1-x_H)=\tau, &\label{constraint03_SPM}\\
& \;\; 0\leq x_L\leq \frac{1}{2}\leq x_H \leq 1,
&\label{last_constraint_SPM}
\end{align}
\end{subequations}
where the Objective (\ref{obj_best_tau_SPM}) is obtained by substituting (\ref{uniqueness_SPM_NE}) into the social welfare in (\ref{sw_multi_path_new}), Constraints (\ref{constraint01_SPM}) and (\ref{constraint03_SPM}) ensures 
 IR and IC as discussed above in Lemma~\ref{condition_both_ir_ic}, respectively, and Constraint (\ref{last_constraint_SPM}) is a result of Lemma \ref{monotonicity_property} and Fig.~\ref{fig02}.

Since $\tau=\frac{c_1}{2}$ is always a feasible solution to Problem (\ref{obj_best_tau_SPM})-(\ref{last_constraint_SPM}), we establish a lower bound of $\frac{1}{2}$ on the PoA of the ASP mechanism in Proposition~\ref{thm_poa_spa_hetero}, where we use $\tau=\frac{c_1}{2}$ to bridge the ASP and the social optimum $SW^*$.
\begin{proposition}\label{thm_poa_spa_hetero}
In two-path routing games with two heterogeneous user types, the ASP mechanism with the optimal threshold $\tau^*$ output by Problem (\ref{obj_best_tau_SPM})-(\ref{last_constraint_SPM}) satisfies IR, IC, and BB, and guarantees a PoA of $\frac{1}{2}$.
 \end{proposition}
\subsection{General Case: Multi-path Network}

Intuitively, achieving high social welfare in (\ref{sw_multi_path_new}) requires users to not only gain higher utility from a few paths with top utility weights (target 1) but also contribute to more diverse PoI aggregation by exploring a broader range of paths (target 2). Building on this insight, we extend the ASP mechanism from Mechanism~\ref{two-path-m} to the following Mechanism~\ref{m_a_spm} for a general multi-path routing game with $k\geq 2$ paths and $m\geq 2$ user types. This extension employs a two-tiered side-payment scheme designed to balance the two targets and accordingly classify paths into two groups, one for rewarding and one for charging.

\begin{mechanism}\label{m_a_spm}
The ASP mechanism for general multi-path routing games operates through two tiers of side payments: 
\begin{itemize}
    \item \textbf{Tier one}:  
The ASP first applies side payments to the $k$ paths via its subroutine, Algorithm~\ref{ALG_pre_side_payment}, carefully charging or rewarding users to rearrange the effective costs of the paths $(P_1, \ldots, P_k)$ in non-increasing order of their utility weights. The resulting cost vector, combining the original costs with the applied charges or rewards, is denoted (with slight abuse of notation) as $(c_1, \ldots, c_k)$. 
    
\item \textbf{Tier two}: 
Based on the cost vector $(c_1, \ldots, c_k)$ resulting from tier-one side payment, the ASP determines a threshold $\theta$ to partition users into two groups for further charges and rewards. Specifically, for some  $\omega\in\{1,...,k\}$ such that  $c_{\omega}>\theta>c_{\omega+1}$, the ASP charges each user on path $P_j\in\{P_{\omega+1},..,P_k\}$ a payment of $(\tau+c_{\omega+1}-c_j)$.
The total payment collected is given by
\begin{equation*}  \Gamma\triangleq\sum\limits_{j=\omega+1}^kx_j\cdot (\tau+c_{\omega+1}-c_j)
\end{equation*}
and is redistributed as rewards evenly to users on paths $\{P_1,...,P_{\omega}\}$, which have costs higher than $\theta$.
In other words, a user on each path $P_j\in\{P_1,...,P_{\omega}\}$ will receive a reward of $(c_j\Gamma)/(x_j\sum\limits_{l=1}^{\omega}c_l)$.
\end{itemize}
\end{mechanism}

In Mechanism~\ref{m_a_spm}, the cost adjustment in tier one serves two key purposes: \textit{first}, it realigns the paths’ costs to match the order of their utility weights, that is, the path with the $j$th largest utility weight is associated with the $j$th largest cost among the original costs. This rearrangement enables a tractable PoA analysis for the ASP; \textit{second}, the subroutine Algorithm~\ref{ALG_pre_side_payment} ensures budget balance in the following way: during each iteration of the “for-loop,” the adjustment applied to path $P_j$ is $(c_j’ - c_j)$, where $c_j$ is the original cost associated with path $P_j$, while $c_j’$ denote the cost that the subroutine assigns to path $P_j$ and is given by the $j$th largest value among those all those original costs in $\mathbf{c}$. By summing up all cost adjustments over the $k$ paths, we find
\begin{equation*}
    \sum\limits_{{P}_j\in {({P}_1,...,{P}_k)}}{(c_j'-c_j)}=   \sum\limits_{{P}_j\in {({P}_1,...,{P}_k)}}c_j'- \sum\limits_{{P}_j\in {({P}_1,...,{P}_k)}}c_j = 0,
\end{equation*}
telling that the total charge equals the total reward in Algorithm~\ref{ALG_pre_side_payment}, thereby satisfying the budget-balance property.

Crucially, the tier-one side-payment scheme in Algorithm~\ref{ALG_pre_side_payment} systematically rearranges the path costs such that paths with higher utility weights are assigned correspondingly higher costs, which is a prerequisite for the subsequent side-payment design in tier two. 

Since path costs are negligible relative to their utilities (i.e., under the assumption $\frac{c}{U(1)} \rightarrow 1$), it is natural to expect that incentivizing a certain portion of users in the equilibrium flow to shift toward paths with higher utility weights can enhance overall social welfare in (\ref{sw_multi_path_new}).  However, over-concentrating users on only a few top-utility paths may reduce the diversity of the final PoI aggregation, thereby compromising the welfare outcome in (\ref{sw_multi_path_new}).
To strike a balance between higher path utility and higher PoI diversity, we introduce in our tier-two side-payment design the parameter $\theta$, which sets the threshold to classify paths into two groups for charging and rewarding, respectively. Furthermore, the second parameter $\tau$ introduced in tier two indeed incentivizes users who were originally inclined to choose low-utility (charged) paths to instead shift toward high-utility (rewarded) paths.

To optimize these two thresholds/parameters $(\tau^*,\theta^*)$, we formulate  Problem \RomanNumeralCaps{2} in (\ref{obj_multiple_dspproblem}) and  (\ref{multi_aspm_ne_flow_general}). In this formulation, the Objective (\ref{obj_multiple_dspproblem}) robustly maximizes the worst-case social welfare under an equilibrium flow of the ASP. Constraint (\ref{multi_aspm_ne_flow_general}) ensures that the ASP achieves an equilibrium flow, derived from
the generalized condition in (\ref{utility_spm_multi_001}), which extends the equilibrium criterion in (\ref{ef_homogeneous}) to the multi-path case.
\begin{equation}\label{utility_spm_multi_001}
\Theta_i(P_j,\widehat{x}_1,...,\widehat{x}_k)=\Theta_i(P_{j+1},\widehat{x}_1,...,\widehat{x}_k),
\end{equation}
\begin{table*}
\begin{subequations}
\begin{align}
\mbox{\textbf{Problem \RomanNumeralCaps{2}:}}&\quad (\tau^*,\theta^*)=\arg\max\limits_{\{\tau\geq 0,\theta\in [\min\limits_{j\in[k]}c_j,\max\limits_{j\in[k]}c_j]\}} \min\limits_{\{\widehat{x}_{\omega+1},...,\widehat{x}_{k}\in [0,1]\bigl| \sum\limits_{i=1}^{\omega} \widehat{x}_i+\sum\limits_{j=\omega+1}^k \widehat{x}_j=1\}} \sum\limits_{i=1}^m \sum\limits_{l=1}^{k}\eta_i w_{il}U(\widehat{x}_l)-c_{\omega+1}-\tau \label{obj_multiple_dspproblem} 
\\
\mbox{\textbf{s.t.}}&\quad
\widehat{x}_i=\frac{c_i\cdot\sum\limits_{j=\omega+1}^k\widehat{x}_j\cdot (\tau+c_{\omega+1}-c_j)}{(c_i-c_{\omega+1}-\tau)\cdot\sum\limits_{l=1}^{\omega}c_l}, \forall{i\in\{1,...,\omega\}} {\rm\; where\;}\omega\triangleq \arg\min\limits_{c_t\geq \theta}\{c_t\},
\label{multi_aspm_ne_flow_general}
\end{align}
\end{subequations}
\end{table*}
%
\begin{algorithm}[t]
\caption{First-tier side-payment for ASP}\label{ALG_pre_side_payment}
\SetKwInOut{Input}{Input}
\SetKwInOut{Output}{Output}
\Input{${\bm c}=\{c_1,c_2,...,c_k\}$, $\mathcal{W}=\{{\bm W}_1,...,{\bm W}_m\}$ wherein each ${\bm W}_i=\{W_{i1},...,W_{ik}\}$, and ${\bm \eta}=\{\eta_1,...,\eta_m\}$.}
\Output{charges and rewards for cost rearrangements.}
Sort paths $P_j\in \mathcal{P}$ in non-increasing order of their utility weights (i.e.,  $[\sum\limits_{i=1}^m \eta_iw_{ij}]$), resulting in sequence $({P}_1,...,{P}_k)$ with their respective costs $(c_1,...,c_k)$.\\
Sort costs $\{c_1,c_2,...,c_k\}$ in non-increasing order resulting in sequence $(c_1',c_2',...,c_k')$.\\
\ForAll{${P}_j\in ({P}_1,...,{P}_k)$}{
\eIf{$c_j>c_j'$}{
 $c_j\leftarrow c_j-(c_j-c_j')$, i.e., charge a user along path ${P}_j$ a payment of ($c_j-c_j'$) as an indirect penalty.
 }
{$c_j\leftarrow c_j+(c_j'-c_j)$, i.e., reward a user along path ${P}_j$ a payment of ($c_j'-c_j$) as an incentive. 
}
}
\end{algorithm}

Due to its Line~1 of path rearrangement, Algorithm~\ref{ALG_pre_side_payment} ensures that path $P_1$ in Mechanism~\ref{m_a_spm} attains the maximum utility weight among all paths. Thereby, one may expect to steer a certain portion of users in an equilibrium flow towards $P_1$ for better social welfare. Indeed, our following Theorem~\ref{new_poa_for_spm} confirms this intuition. As implied by Constraint (\ref{multi_aspm_ne_flow_general}), any equilibrium flow on a charged path must contribute a positive portion of users to the flow on path $P_1$. Leveraging this insight, our Mechanism~\ref{m_a_spm} allocates rewards exclusively to users on path $P_1$, encouraging higher participation along this high-utility path. This redistribution is governed by two thresholds, as formalized in Theorem~\ref{new_poa_for_spm}. 
\begin{theorem}\label{new_poa_for_spm}
In general multi-path routing games, Mechanism~\ref{m_a_spm} (ASP) with its decision thresholds $\theta^* = \frac{c_1+c_2}{2}$ and
\begin{equation}\label{new_tau}
    \tau^* = \left\{\begin{matrix}
0, &{\rm if\;}c_1+c_3-2c_2\leq 0 \\ 
 \frac{c_1+c_3}{2}-c_2,&{\rm if\;}c_1+c_3-2c_2>0 
\end{matrix}\right.
\end{equation}
satisfies IC, IR, and BB, and guarantees a PoA of at least $\frac{1}{2}$ in $O(k\log k)$ time.
\end{theorem}
\begin{proof}
\textbf{Time complexity}: the running time of  Mechanism~\ref{m_a_spm} with $\theta^* = \frac{c_1+c_2}{2}$ and $\tau^*$ in (\ref{new_tau}) is dominated by Algorithm~\ref{ALG_pre_side_payment} for the first-tier side-payment, wherein the first two sorting steps runs in $O(k \log k)$ time, respectively, and each iteration runs in $O(1)$. In total, ASP runs in $O(k\log k)$ time. \textbf{Desiable properties}: By applying a similar analysis as in Proposition~\ref{thm_poa_spa_hetero}, one can easily find that ASP satisfies IR, IC and BB. The following of this proof focuses on \textbf{PoA efficiency}: First, $\theta^* = \frac{c_1+c_2}{2}\in(c_1,c_2)$ implies $\omega=1$ since $c_1>\theta^*>c_2$. Constraint (\ref{multi_aspm_ne_flow_general}) in Problem \RomanNumeralCaps{2} is specified as:
\begin{equation}\label{new_NE_conditionfor_sap}
\widehat{x}_1=\frac{\sum\limits_{j=2}^k\widehat{x}_j\cdot (\tau^*+c_{2}-c_j)}{c_1-c_{2}-\tau^*}.
    \end{equation}
To facilitate our subsequent proof, we give the following Lemma~\ref{prop_for_new_spa_poa} which reveals some important relations .
\begin{lemma}\label{prop_for_new_spa_poa}
By choosing $(\tau^*,\theta^*)$ as in Theorem~\ref{new_poa_for_spm}, it holds for any $j\in\{2,...,k\}$ that $\tau^*+c_2-c_j>c_1-c_2-\tau^*$ and $c_1-c_2-\tau^*>0$.
\end{lemma}
\begin{proof}[Proof of Lemma~\ref{prop_for_new_spa_poa}]
We discuss two cases.

\noindent \textbf{Case 1}. ($c_1+c_3-2c_2\leq 0$). We have $\tau^*=0$ according to (\ref{new_tau}). Hence, $c_1-c_2-\tau^*>0$. For any $j\in\{2,...k\}$, we get
\begin{equation*}
\begin{split}
    \tau^*+c_2-c_j&>\tau^*+c_2-c_3=c_2-c_3>c_1-c_2,
\end{split}
\end{equation*}
in which the first inequality holds because $c_{j-1}>c_j$ for each $j\in\{2,...,k\}$, the equation holds by 
$\tau^*=0$, and the last inequality follows from the base condition $c_1+c_3-2c_2\leq 0$ of this case.

\noindent \textbf{Case 2}. ($c_1+c_3-2c_2> 0$).
According to (\ref{new_tau}), we now have $\tau^* = \frac{c_1+c_3-2c_2}{2}$, which implies by its reorganization that
\begin{equation}\label{basic_case2_lemma9}
    c_1-c_2-\tau^* = \frac{c_1-c_3}{2}>0,
\end{equation}
where the inequality holds as the subroutine Algorithm~\ref{ALG_pre_side_payment} of Mechanism~\ref{m_a_spm} outputs $c_1>c_3$. Therefore, for any $j\in\{2,...,k\}$, the following holds
$ c_2-c_j+\tau^*>c_2-c_3+\tau^*= \frac{c_1-c_3}{2} = c_1-c_2-\tau^*$. This complete proving Lemma~\ref{prop_for_new_spa_poa}.
\end{proof}
Now, we proceed with proving Theorem~\ref{new_tau}. 
Applying Lemma~\ref{prop_for_new_spa_poa} to (\ref{new_NE_conditionfor_sap}) gives us,
\begin{equation*}
    \widehat{x}_1\geq \frac{\sum\limits_{j=2}^k\widehat{x}_j\cdot (\tau^*+c_{2}-c_3)}{(c_1-c_{2}-\tau^*)}\geq \sum\limits_{j=2}^k\widehat{x}_j,
\end{equation*}
in which the first inequality follows since $c_{j-1}>c_j$ holds for each $j\in\{2,...,k\}$, and the second inequality is due to the first family of inequalities in Lemma~\ref{prop_for_new_spa_poa}.

Since $\sum\limits_{j=1}^kx_j=1$, we have $\widehat{x}_1= 1-\sum\limits_{j=2}^k\widehat{x}_j\geq 1-\widehat{x}_1$, which implies that any NE flow in our mechanism admits 
\begin{equation}\label{ne_flow_on_path_1_sap_new}
    \widehat{x}_1\geq \frac{1}{2}.
\end{equation}
By applying (\ref{ne_flow_on_path_1_sap_new}) to the Objective (\ref{obj_multiple_dspproblem}), we can lower bound the social welfare under ASP as follows:
\begin{equation}\label{warmup_bounding_objective01_asp}
\begin{split}
    &{SW_{ASP}}=\min\limits_{\{\widehat{x}_{2},\widehat{x}_{3},...,\widehat{x}_{k}\}} \Big\{\sum\limits_{i=1}^m \sum\limits_{l=1}^{k}\eta_i w_{il}U(\widehat{x}_l)-c_{2}-\tau^*\Big\}\\
       &\geq  \min\limits_{\{\widehat{x}_{2},\widehat{x}_{3},...,\widehat{x}_{k}\}}\Big\{\sum\limits_{i=1}^m\eta_iw_{i1} U(\widehat{x}_1)+\sum\limits_{l=2}^{k}\sum\limits_{i=1}^m \eta_i w_{il}U(\widehat{x}_l)-c_1\Big\}\\
    &\geq  \min\limits_{\{\widehat{x}_{2},\widehat{x}_{3},...,\widehat{x}_{k}\}}\Big\{\sum\limits_{i=1}^m\eta_iw_{i1} U(\widehat{x}_1)\Big\}\\
    &\geq \sum\limits_{i=1}^m\eta_iw_{i1} U(\frac{1}{2})\geq \frac{\sum\limits_{i=1}^m\eta_iw_{i1}}{2}U(1),
\end{split}
\end{equation}
where the first inequality holds by $ c_2+\tau^*\leq c_2+\frac{c_1+c_3}{2}-c_2=\frac{c_1+c_3}{2}$,
the second inequality holds because $U(1)>>c_1$, the third inequality holds by (\ref{ne_flow_on_path_1_sap_new}), and the last inequality is due to the concavity of function $U(x)$. 

On the social optimum $SW^*$, we find 
\begin{equation}\label{warmup_bounding_objective01_opt}
    \begin{split}
      { SW^*} &= \sum\limits_{i=1}^m\sum\limits_{j=1}^k \eta_iw_{ij}U(x_j)-\sum\limits_{j=1}^kx_jc_j\\
      &<  \sum\limits_{i=1}^m\sum\limits_{j=1}^k \eta_iw_{ij}U(x_j)\\
        &=\sum\limits_{j=1}^k U(x_j)[\sum\limits_{i=1}^m \eta_iw_{ij}]\\
        &\leq \sum\limits_{j=1}^k U(x_j)\cdot \max\limits_{j\in\{1,...,k\}}\{\sum\limits_{i=1}^m\eta_iw_{ij}\}\\
        &\leq \max\limits_{j\in\{1,...,k\}}\{\sum\limits_{i=1}^m\eta_iw_{ij}\} U(1)
    \end{split}
\end{equation}
in which the last inequality is due to the following:
\begin{equation*}
\begin{split}
    \sum\limits_{j=1}^k U(x_j)\leq k\cdot U(\frac{\sum\limits_{j=1}^k x_j}{k})\leq U(\sum\limits_{j=1}^k x_j) = U(1).
\end{split}
\end{equation*}
where the first and the second inequalities are due to Jensen's inequality and $U(x)$'s concavity, respectively, and the equation holds since $\sum\limits_{j=1}^k x_j=1$

By comparing (\ref{warmup_bounding_objective01_asp}) and (\ref{warmup_bounding_objective01_opt}), the PoA of our ASP mechanism in a general multi-path routing game follows:
\begin{equation}
\begin{split}
\frac{{SW_{ASP}}}{SW^*}&\geq \frac{\frac{\sum\limits_{i=1}^m\eta_iw_{i1}}{2}U(1)}{\max\limits_{j\in\{1,...,k\}}\{\sum\limits_{i=1}^m\eta_iw_{ij}\} U(1)}\\
&\geq \frac{\sum\limits_{i=1}^m\eta_iw_{i1}}{2\cdot \max\limits_{j\in\{1,...,k\}}\{\sum\limits_{i=1}^m\eta_iw_{ij}\}}\geq \frac{1}{2}
\end{split}
\end{equation}
where the last inequality holds since Line 1 of our Algorithm~\ref{ALG_pre_side_payment} ensures that path $P_1$ in Mechanism~\ref{m_a_spm} has the largest utility weight among all, i.e., $\sum\limits_{i=1}^m\eta_iw_{i1}= \max\limits_{j\in\{1,...,k\}}\{\sum\limits_{i=1}^m\eta_iw_{ij}\}$.
\end{proof}
Theorem~\ref{new_poa_for_spm} shows that Mechanism~\ref{m_a_spm} (ASP) guarantees at least half of the social optimum in the worst case, while maintaining polynomial-time complexity. This improves upon the earlier Mechanism~\ref{m_irm} (AIR), which offers a worst-case PoA of only $\frac{1}{4}$ with even higher computational cost. Importantly, the PoA guarantee in Theorem~\ref{new_poa_for_spm} holds for any utility function $U(x)$ that is concave and monotonically increasing in user flow $x$. Moreover, given a concrete expression of the utility function $U(x)$, our ASP mechanism can even surpass the $\frac{1}{2}$ PoA bound by further optimizing the parameters $\tau$ and $\theta$ via Problem (\ref{obj_multiple_dspproblem})–(\ref{multi_aspm_ne_flow_general}).
As demonstrated by our experiments using a real-world dataset in the subsequent section, the average-case performance of our Mechanism~\ref{m_a_spm} (ASP) approaches the social optimum more closely than its theoretical guarantee in the above Theorem~\ref{new_poa_for_spm}. 
%
%
\section{Numerical Experiments}\label{sec_simulation}
This section evaluates our mechanisms' empirical performances through experiments using a real-world dataset \cite{Gowalladata} of Points-of-Interest (PoIs) in Los Angeles (LA). 
\begin{figure}[t]
    \centering
    \includegraphics[width=8.5cm]{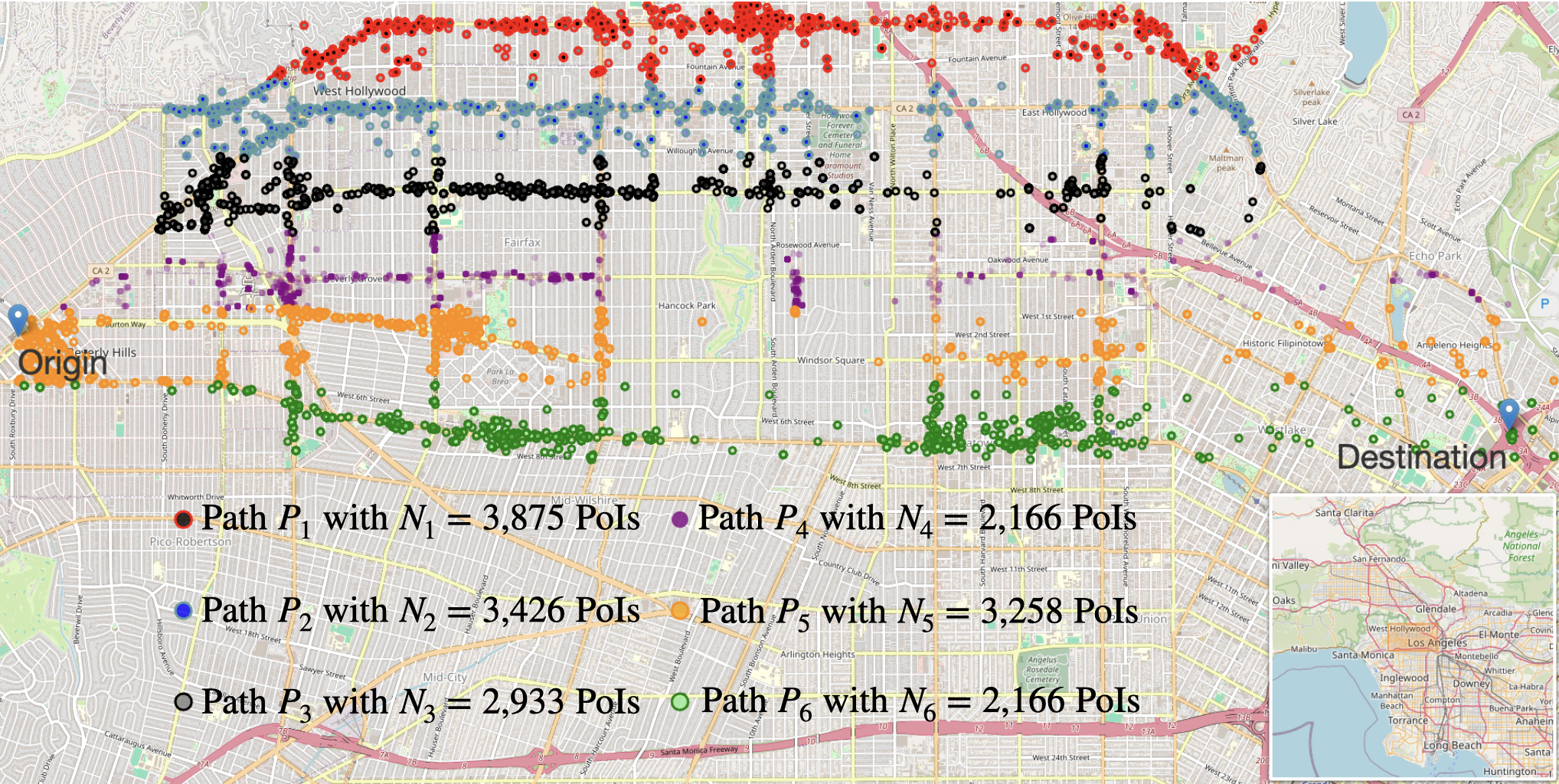}
    \caption{PoI distribution on the target region in Los Angeles specified by a latitude range of [34° 3' 37" N to 34° 5' 59" N] and a longitude range of [118° 14' 28" W to 118° 24' 18" W], which includes 17,724 PoI pieces along 6 parallel pathways for vehicles to harvest from the common origin at Beverley Hills to the downtown destination.}
    \label{poi_map}
\end{figure}

Specifically, our subsequent experiments focus on the social PoI-sharing scenario in Central LA, a historical urban region of LA known for its abundance of PoIs such as landmarks, restaurants, hotels, and attractions. This region spans a latitude range from 34° 3' 37" N to 34° 5' 59" N and a longitude range from 118° 14' 28" W to 118° 24' 18" W, including a total 17,724 of PoIs according to the dataset \cite{Gowalladata}. As visualized in Fig.~\ref{poi_map}, these 17,724 PoIs are distributed across $k=6$ distinct pathways linking a suburb area at Beverley Hills (34° 4' 15" N, 118° 24' 18" W) to a downtown area in Chinatown LA (34° 3' 46" N, 118° 14' 57" W). 
The number of PoIs on each pathway is as follows:  $N_1 = 3,875$ PoIs on path $P_1$ (Sunset Boulevard) in red, $N_2 = 3,426$ PoIs on path $P_2$ (Santa Monica Boulevard) in blue, $N_3 = 2,933$
PoIs on path $P_3$ (the Melrose Boulevard) in black,
$N_4 = 2,116$ PoIs on path $P_4$ (the Beverly Boulevard) in purple, $N_5 = 3,258$ PoIs on path $P_5$ (the West 3rd Street) in orange, and $N_6= 2,116$ PoIs on path $P_6$ (the WilShire Boulevard) in green. 

Given $k=6$ paths, we consider a total number $M = 2,000$ of $m=6$ preference types of users with their type proportions $\eta_i$ and preference proportions $w_{ij}$ to different paths sampled from a random distribution over the normalized interval $(0,1)$, respectively. These parameters adhere to $\sum\limits_{i=1}^m \eta_i = 1$ and $\sum\limits_{j=1}^k w_{ij} =1$ for each user type $i$. To obtain a concrete function $U(x)$ to quantify the amount of PoIs collected by a flow $x$ of users, we suppose that each PoI piece on a path $P_j\in\{P_1,...,P_6\}$ has an equal probability $\frac{1}{N_j}$ to be collected by a user. Consequently, the probability that a PoI piece along path $P_j$ is collected by at least one user of a flow $x$ of $Mx$ users is given by $1-(1-\frac{1}{N_j})^{Mx}$. Thus, a flow $x$ of $Mx$ users can collectively gather an average number $N_j(1-(1-\frac{1}{N_j})^{Mx})$ of PoI pieces, readily giving us the utility function $U_j(x) =N_j\cdot [1-(1-\frac{1}{N_j})^{Mx}]$. This utility function is concave and increasing in $x$ and is also adopted in the literature (e.g., \cite{li2017dynamic,zhang2020efficient}).


 For the multi-path routing game considered in this paper, a feasible mechanism is required to produce an equilibrium flow while satisfying incentive compatibility (IC), individual rationality (IR), and budget balance (BB). To the best of our knowledge, apart from our proposed ASP and AIR mechanisms, no existing mechanisms are known to meet all these criteria. In light of this, we adopt the optimal social welfare as our primary benchmark.
Since finding the social optimum hardly admits a polynomial-time algorithm \cite{li2024survey}, we instead use the maximum social welfare achievable across all possible user flows as an upper bound on the social optimum, serving as our de facto key benchmark for comparison. Denote $N_{\max}=\max\limits_{1\leq j\leq 6} \{N_j\}$ as the maximum number of PoI pieces on a pathway in the experiment. In the following Lemma~\ref{ub_lemma}, we provide an upper bound ($UB$) on the maximum social welfare, i.e., $UB\geq SW^*$ . Further, we apply ratios $\frac{SW_{AIR}}{UB}$ and $\frac{SW_{ASP}}{UB}$ as our evaluation metrics to examine the performances of our mechanisms AIR and ASP, respectively, which indeed provide lower-bound guarantees on our mechanisms' social welfare.
\begin{lemma}\label{ub_lemma}
The optimal social welfare $SW^
*$ for the normalized unit mass of users does not exceed $${  UB}=k\max\limits_{j\in[k]}\Bigl\{\sum\limits_{i=1}^m\eta_iw_{ij}\Bigr\}N_{\max}[1-(1-\frac{1}{N_{\max}})^{\frac{M}{k}}]-\frac{\sum\limits_{j=1}^kc_j}{k}.$$
\end{lemma}

Since our AIR and ASP mechanisms heavily depend on the highest utility weight $\max\limits_{j\in[k]}{\sum\limits_{i=1}^m\eta_iw_{ij}}$ and travel cost among paths, we examine the empirical performances of our AIR and ASP mechanisms by varying the maximum utility weight from 0.3 to 0.75 while the maximum path cost from 10 units to 100 units of gas consumption, respectively.
\begin{figure}[t]
   \centering
 \includegraphics[width = 8.5cm]{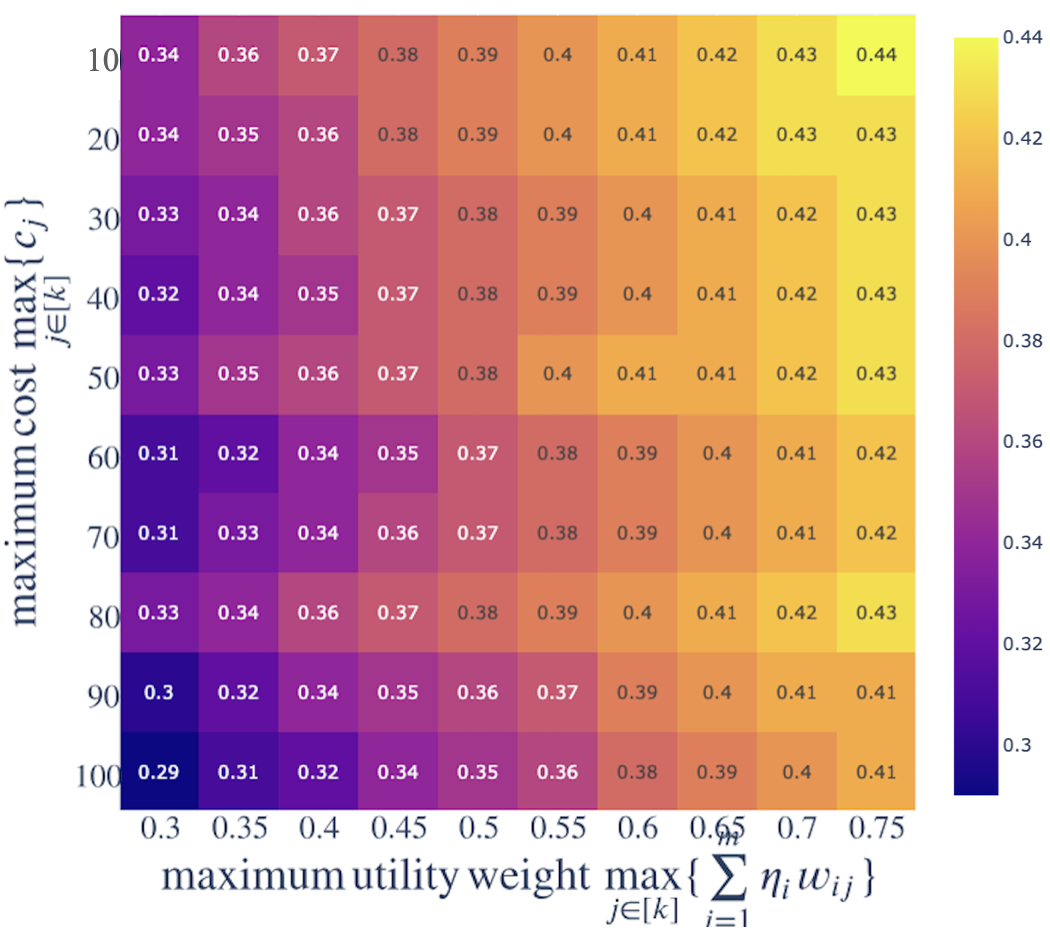}
    \caption{Heatmap for mechanism AIR's experimental results: the average ratio's lower bound $\frac{SW_{AIR}}{UB}$ with \textit{UB} given in Lemma~\ref{ub_lemma}, \textit{versus} the maximum utility weight (i.e., $\max\limits_{j\in[k]}\{\sum\limits_{i=1}^m \eta_iw_{ij}\}$) and the maximum path cost $\max\limits_{j\in[k]}\{c_j\}$.}\label{result_air}
   \end{figure}
\begin{figure}
   \centering
    \includegraphics[width = 8.5cm]{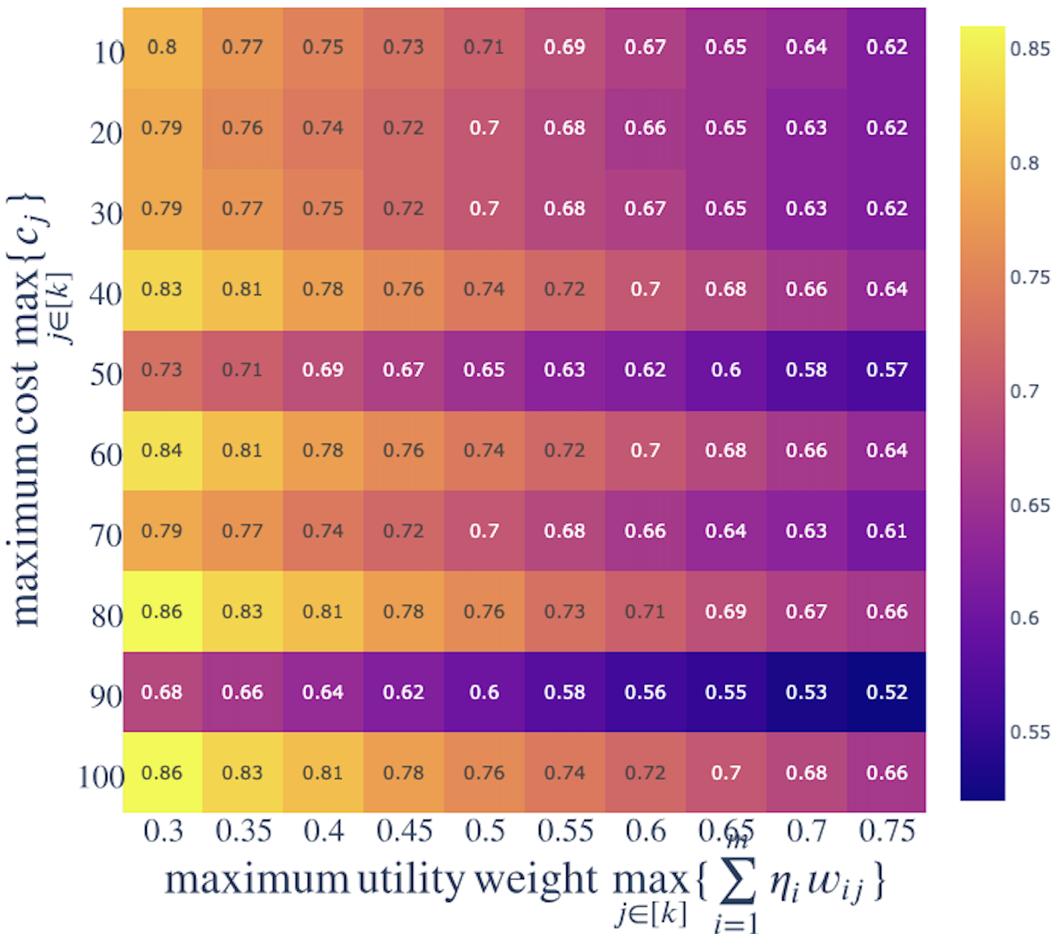}
    \caption{Heatmap for mechanism ASP's experimental results: the average ratio's lower bound $\frac{SW_{ASP}}{UB}$ with \textit{UB} given in Lemma~\ref{ub_lemma}, \textit{versus} the maximum utility weight (i.e., $\max\limits_{j\in[k]}\{\sum\limits_{i=1}^m \eta_iw_{ij}\}$) and the maximum path cost $\max\limits_{j\in[k]}\{c_j\}$.} \label{result_asp}
\end{figure}
%

Heatmaps in Fig.~\ref{result_air} and Fig.~\ref{result_asp} summarize the empirical performances for our mechanisms AIR and ASP \textit{versus} the above $UB$ of the optimal social welfare, respectively. Overall, Fig.~\ref{result_air} and Fig.~\ref{result_asp} show that our AIR and ASP achieve empirical performances significantly better than the worst-case PoA guarantees $\frac{1}{4}$ and $\frac{1}{2}$ as stated in Theorems~\ref{air_multipath_poa} and~\ref{new_poa_for_spm}, respectively. As the maximum utility weight or cost varies, Fig.~\ref{result_air} shows that our AIR practically achieves performances ranging from 29\% to 44\% of the upper bound, and consequently, of the maximum social welfare, while our ASP in Fig.~\ref{result_asp} further improves and achieves near-optimal performance ranging from about 60\% to more about 80\% of the upper bound and also of the maximum social welfare. We note that the ASP mechanism with extra pricing/billing support always outperforms the AIR mechanism in terms of empirical performance. This also corroborates our PoA analysis in Theorems~\ref{air_multipath_poa} and~\ref{new_poa_for_spm}.

Our experimental results presented in Fig.~\ref{result_air} and Fig.~\ref{result_asp} thoroughly analyze the sensitivity of our mechanisms to key parameters, including the maximum path cost $c$ and the maximum utility weight $\sum_{i=1}^m \eta_i w_{ij}$, which indeed have the greatest impact on the performance of our mechanisms. 
The experimental results in Fig.~\ref{result_air} and Fig.~\ref{result_asp} demonstrate the robustness of mechanisms with respect to variations in these parameters. Taking Fig.~\ref{result_asp} of ASP’s experimental results as an example:   When examining a fixed row in the heat map—specifically, the top row labeled $\max\{c\} = 10$, we observe that the ASP mechanism consistently achieves a welfare outcome of at least 62\% of the optimum as the maximum utility weight $\sum_{i=1}^m \eta_i w_{ij}$ varies from $0.3$ to $0.7$. On the other hand, when examining a fixed column—for instance, the leftmost one where $\max\limits_{j\in[k]}\sum_{i=1}^m \eta_i w_{ij} = 0.3$, the results indicate that the ASP mechanism consistently achieves at least 73\% of the optimal social welfare, regardless of how the maximum cost $\max_{j \in [k]}\{c_j\}$ varies within the range $[10, 100]$.

Particularly, our AIR mechanism consistently achieves improved performance as the maximum utility weight increases horizontally in Fig.~\ref{result_air}. This is due to our design in penalty restrictions of AIR in Algorithm~\ref{general_incomplete_air_mechanism}, which incentivizes users to prioritize opting for the path with the largest $\max\limits_{j\in[k]}{\sum\limits_{i\in[m]}\eta_iw_{ij}}$ (i.e., the maximum utility weight of a path). Differently, our ASP mechanism favors a relatively balanced utility weight over paths and performs better for smaller $\max\limits_{j\in[k]}{\sum\limits_{i\in[m]}\eta_iw_{ij}}$ in Fig.~\ref{result_asp}. This is because our design in the side-payments of ASP leads to a balanced equilibrium flow on $k=6$ paths. As $\max\limits_{j\in[k]}{\sum\limits_{i\in[m]}\eta_iw_{ij}}$ decreases and paths' utility weights become balanced, the gap between our social welfare under a guaranteed balanced flow and the maximum social welfare narrows down since our balanced flow can save more utility loss from each pathway. On another note, under a fixed maximum utility weight,
our experimental results also demonstrate that the AIR and ASP mechanisms are both insensitive to path costs even when the cost is not trivial compared to the overall utility.


\begin{figure}[t]
   \centering
 \includegraphics[width = 8cm]{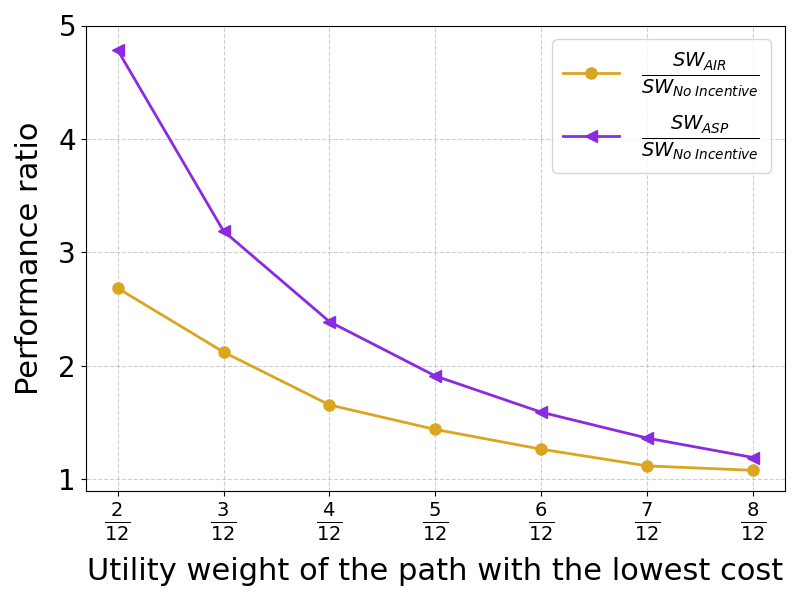}
    \caption{Empirical comparison with the no-incentive benchmark}\label{noincentive_benchmark}
   \end{figure}
   
Besides the above primary experiments comparing our mechanisms with the key benchmark provided in Lemma~\ref{ub_lemma}, we also note that Lemma~\ref{thm_no_incentive_basic} in fact provides an additional benchmark, wherein users, without any incentives, selfishly crowd onto the path with the lowest travel cost.
In view of this, we further compare the empirical performance of our mechanisms against the no-incentive benchmark, which is particularly sensitive to the utility weight of the minimum-cost path. In this experiment, we vary this weight from $\frac{1}{6}$ (representing an average weight) to $\frac{3}{4}$ (indicating a dominant weight). The results are presented in Fig.~\ref{noincentive_benchmark}. Our results show that when the utility weight of the minimum-cost path remains modest, both our ASP and AIR significantly outperform the no-incentive benchmark by up to threefold and fivefold, respectively. As this weight increases to a dominant level, the performance of our mechanisms gradually aligns with that of the no-incentive benchmark. This is because the advantage of PoI diversity diminishes when the utility is overwhelmingly concentrated on a single path.

\section{Concluding Remarks}\label{sec_conclusion}
This paper studies a novel non-atomic routing game with positive network externalities arising from social information sharing, which involves $k\geq 2$ paths and $m\geq 2$ various user types of path preferences. Without any mechanism design, our PoA analysis shows that users' selfish routing choices lead to significantly imbalanced flow distribution, resulting in extremely poor social efficiency with $PoA=0$. To address such huge efficiency loss, we propose both an adaptive information restriction (AIR) mechanism for ease of implementation and an adaptive side-payment mechanism (ASP) for scenarios where transactions for pricing users are further enabled. Our mechanisms protect user privacy by not requiring users' path preferences and satisfy key properties for practical feasibility, including individual rationality (IR), incentive compatibility (IC), and budget balance (BB). Through our rigorous PoA analysis, we show that our AIR and ASP mechanisms both guarantee decent theoretical performances in the worst-case scenarios, in only polynomial time. 
Finally, our theoretical findings are well corroborated by experiments using a real-world dataset. 

\nocite{langley00}
\bibliographystyle{IEEEtran}
\bibliography{mybib}

\begin{thebibliography}{10}
\providecommand{\url}[1]{#1}
\csname url@samestyle\endcsname
\providecommand{\newblock}{\relax}
\providecommand{\bibinfo}[2]{#2}
\providecommand{\BIBentrySTDinterwordspacing}{\spaceskip=0pt\relax}
\providecommand{\BIBentryALTinterwordstretchfactor}{4}
\providecommand{\BIBentryALTinterwordspacing}{\spaceskip=\fontdimen2\font plus
\BIBentryALTinterwordstretchfactor\fontdimen3\font minus \fontdimen4\font\relax}
\providecommand{\BIBforeignlanguage}[2]{{%
\expandafter\ifx\csname l@#1\endcsname\relax
\typeout{** WARNING: IEEEtran.bst: No hyphenation pattern has been}%
\typeout{** loaded for the language `#1'. Using the pattern for}%
\typeout{** the default language instead.}%
\else
\language=\csname l@#1\endcsname
\fi
#2}}
\providecommand{\BIBdecl}{\relax}
\BIBdecl

\bibitem{chang2021mobility}
S.~Chang, E.~Pierson, P.~W. Koh, J.~Gerardin, B.~Redbird, D.~Grusky, and J.~Leskovec, ``Mobility network models of covid-19 explain inequities and inform reopening,'' \emph{Nature}, vol. 589, no. 7840, pp. 82--87, 2021.

\bibitem{vasserman2015implementing}
S.~Vasserman, M.~Feldman, and A.~Hassidim, ``Implementing the wisdom of waze,'' in \emph{Twenty-Fourth International Joint Conference on Artificial Intelligence}, 2015.

\bibitem{cheng2023influence}
X.~Cheng, N.~Li, G.~Rysbayeva, Q.~Yang, and J.~Zhang, ``Influence-aware successive point-of-interest recommendation,'' \emph{World Wide Web}, vol.~26, no.~2, pp. 615--629, 2023.

\bibitem{nikou2014ubiquitous}
S.~Nikou and H.~Bouwman, ``Ubiquitous use of mobile social network services,'' \emph{Telematics and Informatics}, vol.~31, no.~3, pp. 422--433, 2014.

\bibitem{hossain2020ubiquitous}
S.~F.~A. Hossain, Z.~Xi, M.~Nurunnabi, and K.~Hussain, ``Ubiquitous role of social networking in driving m-commerce: evaluating the use of mobile phones for online shopping and payment in the context of trust,'' \emph{Sage Open}, vol.~10, no.~3, p. 2158244020939536, 2020.

\bibitem{garcia2021tutorial}
M.~H.~C. Garcia, A.~Molina-Galan, M.~Boban, J.~Gozalvez, B.~Coll-Perales, T.~{\c{S}}ahin, and A.~Kousaridas, ``A tutorial on 5g nr v2x communications,'' \emph{IEEE Communications Surveys \& Tutorials}, vol.~23, no.~3, pp. 1972--2026, 2021.

\bibitem{chen2017vehicle}
S.~Chen, J.~Hu, Y.~Shi, Y.~Peng, J.~Fang, R.~Zhao, and L.~Zhao, ``Vehicle-to-everything (v2x) services supported by lte-based systems and 5g,'' \emph{IEEE Communications Standards Magazine}, vol.~1, no.~2, pp. 70--76, 2017.

\bibitem{Tripadvisor}
tripadvisor. (2024) Tripadvisor homapage. \url{https://www.tripadvisor.com.sg/}. Accessed: 2024-03-19.

\bibitem{yelp}
Yelp. (2024) Yelp homapage. \url{https://www.yelp.com/}. Accessed: 2024-03-19.

\bibitem{cominetti2022approximation}
R.~Cominetti, M.~Scarsini, M.~Schr{\"o}der, and N.~Stier-Moses, ``Approximation and convergence of large atomic congestion games,'' \emph{Mathematics of Operations Research}, 2022.

\bibitem{li2023congestion}
H.~Li and L.~Duan, ``When congestion games meet mobile crowdsourcing: Selective information disclosure,'' in \emph{Proceedings of the AAAI Conference on Artificial Intelligence}, vol.~37, 2023, pp. 5739--5746.

\bibitem{li2017dynamic}
Y.~Li, C.~A. Courcoubetis, and L.~Duan, ``Dynamic routing for social information sharing,'' \emph{IEEE Journal on Selected Areas in Communications}, vol.~35, no.~3, pp. 571--585, 2017.

\bibitem{macault2022social}
E.~Macault, M.~Scarsini, and T.~Tomala, ``Social learning in nonatomic routing games,'' \emph{Games and Economic Behavior}, vol. 132, pp. 221--233, 2022.

\bibitem{cominetti2019price}
R.~Cominetti, M.~Scarsini, M.~Schr{\"o}der, and N.~E. Stier-Moses, ``Price of anarchy in stochastic atomic congestion games with affine costs,'' in \emph{Proceedings of the 2019 ACM Conference on Economics and Computation}, 2019, pp. 579--580.

\bibitem{kim2021simple}
B.~Kim, S.~W. Kim, S.~M. Iravani, and K.~S. Park, ``Simple mechanisms for sequential capacity allocations,'' \emph{Production and Operations Management}, vol.~30, no.~9, pp. 2925--2943, 2021.

\bibitem{li2018customer}
B.~Li, D.~Hao, D.~Zhao, and T.~Zhou, ``Customer sharing in economic networks with costs,'' in \emph{Proceedings of the 27th International Joint Conference on Artificial Intelligence}, 2018, pp. 368--374.

\bibitem{zhu2022information}
Y.~Zhu and K.~Savla, ``Information design in nonatomic routing games with partial participation: Computation and properties,'' \emph{IEEE Transactions on Control of Network Systems}, vol.~9, no.~2, pp. 613--624, 2022.

\bibitem{li2024human}
H.~Li and L.~Duan, ``Human-in-the-loop learning for dynamic congestion games,'' \emph{IEEE Transactions on Mobile Computing}, 2024.

\bibitem{wu2021value}
M.~Wu, S.~Amin, and A.~E. Ozdaglar, ``Value of information in bayesian routing games,'' \emph{Operations Research}, vol.~69, no.~1, pp. 148--163, 2021.

\bibitem{li2024optimize}
H.~Li and L.~Duan, ``To optimize human-in-the-loop learning in repeated routing games,'' \emph{IEEE Transactions on Mobile Computing}, 2024.

\bibitem{giordano2023note}
F.~Giordano, ``A note on social learning in non-atomic routing games,'' \emph{Operations Research Letters}, vol.~51, no.~3, pp. 259--265, 2023.

\bibitem{li2025analyze}
H.~Li and L.~Duan, ``To analyze and regulate human-in-the-loop learning for congestion games,'' \emph{IEEE Transactions on Networking}, 2025.

\bibitem{tavafoghi2017informational}
H.~Tavafoghi and D.~Teneketzis, ``Informational incentives for congestion games,'' in \emph{2017 55th Annual Allerton Conference on Communication, Control, and Computing (Allerton)}.\hskip 1em plus 0.5em minus 0.4em\relax IEEE, 2017, pp. 1285--1292.

\bibitem{farhadi2022dynamic}
F.~Farhadi and D.~Teneketzis, ``Dynamic information design: A simple problem on optimal sequential information disclosure,'' \emph{Dynamic Games and Applications}, vol.~12, no.~2, pp. 443--484, 2022.

\bibitem{kordonis2019mechanisms}
I.~Kordonis, M.~M. Dessouky, and P.~A. Ioannou, ``Mechanisms for cooperative freight routing: Incentivizing individual participation,'' \emph{IEEE Transactions on Intelligent Transportation Systems}, vol.~21, no.~5, pp. 2155--2166, 2019.

\bibitem{ghafelebashi2023congestion}
A.~Ghafelebashi, M.~Razaviyayn, and M.~Dessouky, ``Congestion reduction via personalized incentives,'' \emph{Transportation Research Part C: Emerging Technologies}, vol. 152, p. 104153, 2023.

\bibitem{mansour2022bayesian}
Y.~Mansour, A.~Slivkins, V.~Syrgkanis, and Z.~S. Wu, ``Bayesian exploration: Incentivizing exploration in bayesian games,'' \emph{Operations Research}, vol.~70, no.~2, pp. 1105--1127, 2022.

\bibitem{gkatzelis2016optimal}
V.~Gkatzelis, K.~Kollias, and T.~Roughgarden, ``Optimal cost-sharing in general resource selection games,'' \emph{Operations Research}, vol.~64, no.~6, pp. 1230--1238, 2016.

\bibitem{christodoulou2024resource}
G.~Christodoulou, V.~Gkatzelis, and A.~Sgouritsa, ``Resource-aware cost-sharing methods for scheduling games,'' \emph{Operations Research}, vol.~72, no.~1, pp. 167--184, 2024.

\bibitem{georgoulaki2021equilibrium}
E.~Georgoulaki, K.~Kollias, and T.~Tamir, ``Equilibrium inefficiency and computation in cost-sharing games in real-time scheduling systems,'' \emph{Algorithms}, vol.~14, no.~4, p. 103, 2021.

\bibitem{gkatzelis2021resource}
V.~Gkatzelis, E.~Pountourakis, and A.~Sgouritsa, ``Resource-aware cost-sharing mechanisms with priors,'' in \emph{Proceedings of the 22nd ACM Conference on Economics and Computation}, 2021, pp. 541--559.

\bibitem{christodoulou2020resource}
G.~Christodoulou, V.~Gkatzelis, M.~Latifian, and A.~Sgouritsa, ``Resource-aware protocols for network cost-sharing games,'' in \emph{Proceedings of the 21st ACM conference on economics and computation}, 2020, pp. 81--107.

\bibitem{papadimitriou2001algorithms}
C.~Papadimitriou, ``Algorithms, games, and the internet,'' in \emph{Proceedings of the thirty-third annual ACM symposium on Theory of computing}, 2001, pp. 749--753.

\bibitem{derrow2021eta}
A.~Derrow-Pinion, J.~She, D.~Wong, O.~Lange, T.~Hester, L.~Perez, M.~Nunkesser, S.~Lee, X.~Guo, B.~Wiltshire \emph{et~al.}, ``Eta prediction with graph neural networks in google maps,'' in \emph{Proceedings of the 30th ACM international conference on information \& knowledge management}, 2021, pp. 3767--3776.

\bibitem{Googlecrowdsourcea}
G.~Crowdsource. (2025) Help make ai serve everyone, everywhere. \url{https://crowdsource.google.com/about.html}. Accessed: 2025-05-13.

\bibitem{roughgarden2002bad}
T.~Roughgarden and {\'E}.~Tardos, ``How bad is selfish routing?'' \emph{Journal of the ACM (JACM)}, vol.~49, no.~2, pp. 236--259, 2002.

\bibitem{richman2007topological}
O.~Richman and N.~Shimkin, ``Topological uniqueness of the nash equilibrium for selfish routing with atomic users,'' \emph{Mathematics of Operations Research}, vol.~32, no.~1, pp. 215--232, 2007.

\bibitem{altman2001routing}
E.~Altman, T.~Basar, T.~Jim{\'e}nez, and N.~Shimkin, ``Routing into two parallel links: Game-theoretic distributed algorithms,'' \emph{Journal of Parallel and Distributed Computing}, vol.~61, no.~9, pp. 1367--1381, 2001.

\bibitem{bergemann2005robust}
D.~Bergemann and S.~Morris, ``Robust mechanism design,'' \emph{Econometrica}, pp. 1771--1813, 2005.

\bibitem{balseiro2019dynamic}
S.~R. Balseiro, O.~Besbes, and G.~Y. Weintraub, ``Dynamic mechanism design with budget-constrained buyers under limited commitment,'' \emph{Operations Research}, vol.~67, no.~3, pp. 711--730, 2019.

\bibitem{dai2018aggregation}
W.~Dai, G.~Jin, J.~Lee, and M.~Luca, ``Aggregation of consumer ratings: an application to yelp. com,'' \emph{Quantitative Marketing and Economics}, vol.~16, pp. 289--339, 2018.

\bibitem{googlemapspoints}
G.~Maps. (2025) Local guides points, levels \& badging. \url{https://maps.google.com/intl/en-GB/localguides.html}. Accessed: 2025-05-13.

\bibitem{wazepoints}
Waze. (2025) Your rank and points. \url{https://www.waze.com/discuss/t/your-rank-and-points/379514.html}. Accessed: 2025-05-13.

\bibitem{zhang2024survey}
J.~Zhang, Y.~Bi, M.~Cheng, J.~Liu, K.~Ren, Q.~Sun, Y.~Wu, Y.~Cao, R.~C. Fernandez, H.~Xu \emph{et~al.}, ``A survey on data markets,'' \emph{arXiv preprint arXiv:2411.07267}, 2024.

\bibitem{Gowalladata}
GowallaData. (2022) Gowalla dataset. \url{https://snap.stanford.edu/data/loc-gowalla.html}. Accessed: 2023-01-13.

\bibitem{zhang2020efficient}
Y.~Zhang, X.~Lan, J.~Ren, and L.~Cai, ``Efficient computing resource sharing for mobile edge-cloud computing networks,'' \emph{IEEE/ACM Transactions on Networking}, vol.~28, no.~3, pp. 1227--1240, 2020.

\bibitem{li2024survey}
H.~Li, W.~Huang, Z.~Duan, D.~H. Mguni, K.~Shao, J.~Wang, and X.~Deng, ``A survey on algorithms for nash equilibria in finite normal-form games,'' \emph{Computer Science Review}, vol.~51, p. 100613, 2024.

\end{thebibliography}


%
%
%
%
%
\end{document}